\documentclass[conference]{IEEEtran}
\IEEEoverridecommandlockouts

\usepackage[utf8]{inputenc}
\usepackage[T1]{fontenc}
\usepackage{stmaryrd}
\usepackage{amsthm, amsmath}
\usepackage{amssymb}
\usepackage{bbm}
\usepackage{mathtools}
\usepackage{thm-restate}
\usepackage{xspace}
\usepackage[textwidth=2cm,textsize=small]{todonotes}
\usepackage{hyperref}
\usepackage{cleveref}
\usepackage{stackrel}
\usepackage{MnSymbol} 
\usepackage{todonotes}
\usepackage{shuffle}
\usepackage{paralist}
\usepackage{graphicx}
\usepackage{etoolbox}
\usepackage{textcase}
\usepackage{enumitem}
\usepackage[bb=boondox]{mathalfa}

\usepackage[contents={}]{background}
\newcommand\BGfrom[1]{%
\AddEverypageHook{%
  \ifnum\value{page}>\numexpr#1-1\relax
    \backgroundsetup{
      contents={Over the limit},
      color=red,
      scale=10,
      opacity=0.5
    }%
  \fi
  \BgMaterial%
  }%
}

\usepackage{amstext} 
\usepackage{array}   
\newcolumntype{C}{>{$}c<{$}} 
\newcolumntype{L}{>{$}l<{$}} 

\hypersetup{
    colorlinks,
    linkcolor={red!50!black},
    citecolor={blue!50!black},
    urlcolor={blue!80!black}
}

\usepackage{tikz}
\usetikzlibrary{positioning}
\usetikzlibrary{decorations.markings}

\makeatletter
\patchcmd{\@sect}{\uppercase}{\MakeTextUppercase}{}{}
\patchcmd{\@sect}{\uppercase}{\MakeTextUppercase}{}{}
\makeatother

\newcommand{\tuple}[1]{\left(#1\right)}

\newcommand{\goesto}[1]{\xrightarrow{#1}}

\newcommand{\e}{\varepsilon}

\newcommand{\ignore}[1]{}

\newcommand{\wlg}{w.l.o.g.}
\newcommand{\Wlg}{W.l.o.g.}
\newcommand{\wrt}{w.r.t.}
\newcommand{\vs}{vs.}
\newcommand{\cf}{cf.}
\newcommand{\eg}{e.g.}
\newcommand{\Eg}{E.g.}
\newcommand{\aka}{a.k.a.}
\newcommand{\ie}{i.e.}
\newcommand{\Ie}{I.e.}
\newcommand{\rhs}{r.h.s.}
\newcommand{\lhs}{l.h.s.}

\newif\ifstartedinmathmode
\newcommand*{\st}{
  \relax\ifmmode\startedinmathmodetrue\else\startedinmathmodefalse\fi
  \ifstartedinmathmode{\;\cdot\;}\else{s.t.}\fi%
}

\newcommand{\A}{\mathbb A}

\newcommand{\C}{\mathbb C}

\newcommand{\N}{\mathbb N}
\newcommand{\Z}{\mathbb Z}
\newcommand{\Q}{\mathbb Q}

\renewcommand{\SS}{\mathcal S}

\newcommand{\FF}{\mathcal F}

\let\oldcirc\circ
\let\circ\relax
\DeclareMathOperator{\circ}{\oldcirc\,}

\let\oldshuffle\shuffle
\let\shuffle\relax
\DeclareMathOperator{\shuffle}{\oldshuffle}

\DeclareMathOperator{\hadamard}{\odot}
\DeclareMathOperator{\infiltration}{\uparrow}

\newcommand{\card}[1]{\left|#1\right|}

\newcommand{\set}[1]{\left\{#1\right\}}
\newcommand{\setof}[2]{\set{#1 \;\middle|\; #2}}

\newcommand{\poly}[2]{#1[#2]}
\newcommand{\polyof}[3]{\poly {#1} {#2 \mid #3}}

\newcommand{\ncpoly}[2]{#1\langle#2\rangle}

\newcommand{\powerseries}[2]{#1[\![#2]\!]}

\newcommand{\series}[2]{#1\langle\!\langle#2\rangle\!\rangle}
\newcommand{\commseries}[2]{#1\langle\!\langle#2\rangle\!\rangle^{\mathsf{comm}}}

\newcommand{\ideal}[1]{\langle#1\rangle}
\newcommand{\idealof}[2]{\ideal{#1 \mid #2}}

\newcommand{\sem}[1]{\left\llbracket#1\right\rrbracket}

\newcommand{\bigO}[1]{O\left(#1\right)}

\newcommand{\multiplicity}[2]{\#_{#1}{#2}}
\newcommand{\Parikh}[1]{\multiplicity {} {#1}}
\newcommand{\coefficient}[2]{[#1]#2}

\newcommand{\inner}[2]{\langle#1, #2\rangle}

\newcommand{\deriveleft}[1]{\delta^L_{#1}}
\newcommand{\deriveright}[1]{\delta^R_{#1}}

\let\oldpartial\partial
\renewcommand{\partial}[2]{\oldpartial_{#1}{#2}}

\newcommand{\shift}[1]{\sigma_{#1}}

\newcommand{\one}{\mathbbm 1}
\newcommand{\zero}{\mathbb 0}

\newcommand{\CDA}{\textsf{CDA}}

\newcommand{\EXPTIME}{\textsf{EXPTIME}}

\newcommand{\EXPSPACE}{\textsf{EXPSPACE}}
\newcommand{\TWOEXPTIME}{{2-\EXPTIME}}

\newcommand{\TWOEXPSPACE}{{2-\EXPSPACE}}

\newcolumntype{A}{
 >{$}r<{$}
 @{\extracolsep{0pt}}
 >{${}} l <{$}
} 

\newcolumntype{M}{
 >{$}c<{$}
}

\theoremstyle{plain}
\newtheorem{theorem}{Theorem}
\newtheorem*{theorem*}{Theorem}
\newtheorem{lemma}{Lemma}

\newtheorem{claim}{Claim}
\newtheorem*{claim*}{Claim}

\newtheorem*{proposition*}{Proposition}

\newtheorem*{fact*}{Fact}

\theoremstyle{remark}
\newtheorem{remark}{Remark}
\newtheorem{example}{Example}

\theoremstyle{definition}
\newtheorem{definition}{Definition}

\theoremstyle{problem}
\newtheorem{problem}{Problem}

\newenvironment{claimproof}
  {\begin{proof}[Proof of the claim]}
    {\end{proof}}

\crefname{equation}{}{}
\crefname{definition}{Definition}{Definitions}
\crefname{claim}{Claim}{Claims}
\crefname{proposition}{Proposition}{Propositions}
\crefname{lemma}{Lemma}{Lemmas}
\crefname{corollary}{Corollary}{Corollaries}
\crefname{theorem}{Theorem}{Theorems}
\crefname{figure}{Figure}{Figures}
\crefname{fact}{Fact}{Facts}
\crefname{example}{Example}{Examples}
\crefname{remark}{Remark}{Remarks}
\crefname{section}{§}{§§}
\crefname{subsection}{§}{§§}
\crefname{subsubsection}{§}{§§}
\crefname{enumi}{}{}
\crefname{problem}{Problem}{Problems}

\begin{document}

\title{The commutativity problem for effective varieties of formal series, and applications
\thanks{Partially supported by the ERC grant INFSYS, agreement no. 950398,
and the Polish National Science Centre grant no. 2024/54/E/ST6/00287.}
}

\author{\IEEEauthorblockN{}
}

\author{\IEEEauthorblockN{Lorenzo Clemente}
    \IEEEauthorblockA{\textit{Department of Informatics} \\
    \textit{University of Warsaw} \\
    Warszawa, Poland \\
    \href{https://orcid.org/0000-0003-0578-9103}}
}

\maketitle

\begin{abstract}
    A \emph{formal series} in noncommuting variables $\Sigma$ over the rationals is a mapping $\Sigma^* \to \Q$
    from the free monoid generated by $\Sigma$ to the rationals.
    We say that a series is \emph{commutative} if the value in the output
    does not depend on the order of the symbols in the input.
    The \emph{commutativity problem} for a class of series
    takes as input a (finite presentation of) a series from the class
    and amounts to establishing whether it is commutative.
    This is a very natural, albeit nontrivial problem,
    which has not been considered before from an algorithmic perspective.
    %
    
    We show that commutativity is decidable
    for all classes of series that constitute a so-called \emph{effective prevariety},
    a notion generalising Reutenauer's varieties of formal series.
    For example, the class of \emph{rational series},
    introduced by Schützenberger in the 1960's,
    is well-known to be an effective (pre)variety,
    and thus commutativity is decidable for it.

    In order to showcase the applicability of our result,
    we consider classes of formal series generalising the rational ones.
    We consider \emph{polynomial automata},
    \emph{shuffle automata},
    and \emph{infiltration automata},
    and we show that each of these models
    recognises an effective prevariety of formal series.
    Consequently, their commutativity problem is decidable,
    which is a novel result.
    We find it remarkable that commutativity can be decided in a uniform way
    for such disparate computation models.
    
    Finally, we present applications of commutativity outside the theory of formal series.
    We show that we can decide \emph{solvability} in sequences and in power series
    for restricted classes of algebraic difference and differential equations,
    for which such problems are undecidable in full generality.
    Thanks to this, we can prove that the syntaxes of multivariate \emph{polynomial recursive sequences}
    and of \emph{constructible differentially algebraic power series} are effective,
    which are new results which were left open in previous work. 
    %
\end{abstract}


\section*{Introduction}

A \emph{univariate formal series} over the rationals is a function $\N \to \Q$
(sometimes written $\powerseries \Q x$).
They arise in a variety of disciplines, \eg,
Taylor series in mathematical analysis~\cite{Rudin:Analysis:1976},
generating functions in combinatorics~\cite{Stanley:1986,Wilf:generatingfunctionology:2005},
and Z-transforms in engineering~\cite{OppenheimSchafer:DSP:1975}.
The denomination ``formal'' highlights the algebraic aspect, rather than the analytic one;
for brevity, we will refer to formal series simply as series.

The theory of univariate series can be developed in several directions.
One direction is more popular in mathematics,
and considers \emph{multivariate series in commuting indeterminates $X = \set{x_1, \dots, x_d}$},
which are functions of the form $\N^d \to \Q$ for some $d \in \N_{\geq 1}$
(sometimes written $\powerseries \Q X$)~\cite{PemantleWilsonMelczer:CUP:2024}.
Another direction is more popular in theoretical computer science,
and considers \emph{multivariate series in noncommuting indeterminates $\Sigma = \set{a_1, \dots, a_d}$},
which are functions of the form $\Sigma^* \to \Q$
(sometimes called \emph{weighted languages} and written $\series \Q \Sigma$)~\cite{BerstelReutenauer:CUP:2010}.
In fact, noncommuting indeterminates are more general,
since we can recover commuting indeterminates as a special case:
Namely, as those series $f \in \series \Q \Sigma$ which are \emph{commutative},
\ie, $f(u) = f(v)$ for all words $u, v \in \Sigma^*$ which are permutations of each other.
This is the unifying view that we shall adopt.

In the context of formal languages and automata theory,
series in noncommuting indeterminates have been introduced by Schützenberger,
who considered the class of \emph{rational series}~\cite{Schutzenberger:IC:1961}
(\cf~\cite{BerstelReutenauer:CUP:2010} for an extensive introduction).
This class has since appeared in various areas of computer science and mathematics~\cite{HandbookWA},
such as speech recognition~\cite{MohriPereira:Riley:CSL:2002},
digital image compression~\cite[Part IV]{HandbookWA},
learning~\cite{BalleMohri:AI:2015},
quantum computing~\cite{BlondelJeandelKoiranPortier:SIAM:JoC:2005},
as well as in Bell and Smertnig's resolution~\cite{BellSmertnig:SM:2021}
of a conjecture of Reutenauer~\cite{Reutenauer:FCT:1979}.

Series have also been applied to control theory,
as put forward by the pioneering work of Fliess~\cite{Fliess:1981}.
Namely, he associated to a control system its \emph{characteristic series},
in such a way that the behaviour of the system as an input-output transformer
is uniquely determined by it~\cite[Lemme II.1]{Fliess:1981}.
He then recovered the rational series
as the characteristic series of the so-called \emph{bilinear systems}~\cite{Fliess:MST:1976}.
This is an extensively studied class~\cite{DAlessandroIsidoriRuberti:SIAMJoC:1974},
whose theory can be derived from well-known results for rational series
(\eg, decidability of equivalence checking and existence of minimal realisations).

It is in relation with control theory
that commutativity naturally enters the picture.
Fliess isolated an important subclass of bilinear systems,
namely those 
where the output at time $t$ depends only on $t$ and the inputs at time $t$ (but not at times $< t$),
and proved that their characteristic series coincide with the \emph{commutative} rational series~\cite[Proposition 3.6]{Fliess:MST:1976}.

Commutativity, besides arguably being a fundamental algebraic notion per se,
also appears in a number of other works.
For instance, it is used as a simplifying assumption
in Karhumäki's study on \emph{$\N$-rational series}~\cite{Karhumaki:TCS:1977}
(later corrected in~\cite{Lopez:STACS:2025}),
it is featured in Litow's short note on the undecidability of the multivariate \emph{Skolem problem}~\cite{Litow:IPL:2003},
in the study of \emph{DT0L systems}~\cite{Karhumaki:TCS:1979},
in quantum finite automata~\cite{HuangCao:CTISC:2021},
and in the formal verification of MapReduce frameworks~\cite{ChenSongWu:CAV:2016}.
Under the assumption of commutativity,
an elementary upper bound for the computation of \emph{algebraic invariants} of weighted finite automata
has been recently obtained~\cite{NosanAmaurySchmitzShirmohammadiWorrell:ISSAC:2022}
(the complexity of the general problem being open~\cite{HrushovskiOuakninePoulyWorrell:LICS:2018,HrushovskiOuakninePouly:Worrell:JACM:2023}).

All these works leveraging commutativity offer no indication
as whether one can decide it algorithmically. 
Moreover, most decision problems for rational series are undecidable
\eg, \emph{containment} and \emph{nonemptiness}~\cite[Theorem 21]{NasuHonda:IC:1969}\-
(\cf~also~\cite[Theorem 6.17]{Paz:1971}),
therefore finding problems which are decidable for the whole class is challenging.

\subsection*{Contributions}

We study the \emph{commutativity problem} for classes of series from a computational perspective.
For rational series, thanks to strong structural properties~\cite[Proposition 3.1]{Fliess:1971-1972}
commutativity reduces to whether the linear maps in a \emph{minimal linear representation} commute pairwise.
Thanks to Schützenberger's seminal result,
such minimal representations can be computed efficiently\-
~\cite[III.B.1.]{Schutzenberger:IC:1961},
and thus we can check commutativity in polynomial time by a direct numerical computation.
This is in line with the fact that decidable problems on rational series
often rely on minimisation in an essential way\-
\cite{BellSmertnig:LICS:2023,Kostolanyi:IC:2023,JeckerMazowieckiPurser:LICS:2024}.
For more general classes where minimisation is not available,
commutativity is an open problem.

\subsubsection*{First contribution}
%
We show that commutativity is decidable for \emph{effective prevarieties} of series,
a generalisation of Reutenauer's \emph{varieties}~\cite[Sec.~III.1]{Reutenauer_1980}.
\emph{Prevarieties} are classes of series closed under linear combinations,
right, and left derivatives (\cf~\cref{sec:preliminaries}).
\emph{Effectiveness} demands that series in the class can be manipulated by algorithms
working on finite representations.
This achieves a delicate balance:
It is sufficiently strong to ensure decidability of commutativity,
but sufficiently weak to be applicable to several classes of series
not known to have minimal representations.
\begin{restatable}{theorem}{thmCommutativityForPrevarieties}
    \label{thm:commutativity for effective prevarieties}
    Let $\SS$ be an effective prevariety of series.
    The commutativity problem for $\SS$ is decidable.
\end{restatable}
\noindent
This gives us a uniform framework
in which many classes of series can be shown to have decidable commutativity.
For example, rational series are effective (pre)varieties~\cite[Theorem III.1.1]{Reutenauer_1980},
giving an alternative decidability proof.

\subsubsection*{Second contribution}
We apply~\cref{thm:commutativity for effective prevarieties} to decide commutativity
for three, mutually incomparable classes of series, generalising the rational ones.
Namely, we consider series recognised by \emph{polynomial automata}\-
\cite{BenediktDuffSharadWorrell:PolyAut:2017}~(\cref{sec:hadamard}),
\emph{shuffle automata}~\cite{Clemente:CONCUR:2024}~(\cref{sec:shuffle}),
and \emph{infiltration automata}~(\cref{sec:infiltration}).
%
%
Polynomial automata (also known as \emph{cost register automata} over the rationals\-
\cite{AlurDAntoniDeshmukhRaghothamanYuan:LICS:2013})
and shuffle automata are known to have a decidable equality problem\-
\cite{BenediktDuffSharadWorrell:PolyAut:2017,Clemente:CONCUR:2024},
however they have very different computational properties,
\eg, Ackermannian for the former while elementary for the latter.
They were not known to be prevarieties.
The notion of infiltration automata is novel,
as well as decidability of their equality problem.
Thanks to the notion of effective prevariety,
we can decide commutativity for these classes in a common framework,
which we find remarkable.
As far as we know, no other natural problem beyond equality
was known to be decidable for these classes.
\begin{theorem}
    \label{thm:main theorem}
    Polynomial automata, shuffle automata, and infiltration automata
    recognise effective prevarieties.
    As a consequence, the commutativity problem is decidable for them.
\end{theorem}
\noindent
We reduce commutativity to checking finitely many equalities (\cref{lem:commutativity}).
In turn, we show that equality reduces to commutativity,
yielding a complete picture.

\begin{example}
    \label{ex:intro:1}
    As an application of~\cref{thm:main theorem},
    consider a polynomial automaton over alphabet $\Sigma = \set{a_1, a_2}$ with one register $r$.
    Initially, the register starts at some rational value $r := c \in \Q$.
    Upon reading $a_1$, it is updated with the rule $r := r^2$
    and upon reading $a_2$ with $r := 1 - r^2$.
    After reading the input, the automaton outputs the value of $r$.
    For instance, reading $a_1 a_2$, the automaton outputs $1 - c^4$,
    while over $a_2 a_1$ it outputs $(1 - c^2)^2$.
    For $c = -1$, these two outputs coincide,
    and in fact our algorithms determines that
    the semantics of the automaton is commutative for this initial value.
    For other values, \eg, $c = 2$, the semantics is not commutative:
    $a_1 a_2 \mapsto 1 - 2^4 = -15$,
    but $a_2 a_1 \mapsto (1 - 2^2)^2 = 9$.
\end{example}

\subsubsection*{Third contribution}
We apply \cref{thm:main theorem}
to two problems on series in commuting variables.
We find it intriguing that we can solve problems about series in commuting variables
using an algorithmic result in the theory of noncommuting variables.
This corroborates the hypothesis that the latter
are a more fundamental object than the former.

\paragraph{Multivariate polyrec sequences}
In the first problem, we propose a notion of \emph{multivariate polynomial recursive sequences} (\emph{polyrec}),
which are a subset of $\N^d \to \Q$ generalising the univariate polyrec sequences
from~\cite{Cadilhac:Mazowiecki:Paperman:Pilipczuk:Senizergues:ToCS:2021}.
Multivariate polyrec sequences are defined via systems of \emph{polynomial difference equations} of a certain form
(\cf~\cref{sec:multivariate polyrec} for a precise definition).
However, such systems may not have any solution at all,
and it is not clear whether this can be decided.
Solvability of polynomial difference equations is in general undecidable
(\cref{rem:undecidability of polynomial difference equations}),
which poses an obstacle to the development of a computational theory of multivariate polyrec sequences.
We overcome this difficulty by showing that commutativity of polynomial automata
can be used to decide whether polyrec systems are solvable~(\cref{thm:polyrec consistency}).
Thus, the syntax of multivariate polyrec sequences is effective.

\begin{example}
    \label{ex:intro:2}
    Consider the bivariate polyrec recursions
    $f(n_1 + 1, n_2) := f(n_1, n_2)^2$ and
    $f(n_1, n_2 + 1) := 1 - f(n_1, n_2)^2$, for all $n_1, n_2 \in \N$.
    As shown in~\cref{fig:polyrec constraints},
    $f(2, 2)$ can be computed in six possible ways,
    all of which must give the same value.
    By reducing to commutativity for the polynomial automaton of~\cref{ex:intro:1},
    our algorithm determines the existence of
    a sequence solution $f : \N^2 \to \Q$ starting at $f(0, 0) = -1$,
\end{example}

\newcommand{\thegrid}{
    \draw[step=7mm] (0,0) grid (2,2);
    \foreach \i in {0,...,2}
        \foreach \j in {0,...,2}
            \fill (\i,\j) circle (1pt);
}

\begin{figure}
    \begin{center}
        \tikzset{x=7mm, y=7mm, every node/.append style={font=\footnotesize}}
        \begin{tabular}{ccc}
            \begin{tikzpicture}
                \thegrid
                \draw[ultra thick, blue, ->] (0,0) -- (0,1);
                \draw[ultra thick, blue, ->] (0,1) -- (0,2);
                \draw[ultra thick, blue, ->] (0,2) -- (1,2);
                \draw[ultra thick, blue, ->] (1,2) -- (2,2);
                \fill[black]
                    (0,0) circle (2pt) node[below] {$f(0,0)$}
                    (2,2) circle (2pt) node[above] {$f(2,2)$};
            \end{tikzpicture}&
            \begin{tikzpicture}
                \thegrid
                \draw[ultra thick, blue, ->] (0,0) -- (0,1);
                \draw[ultra thick, blue, ->] (0,1) -- (1,1);
                \draw[ultra thick, blue, ->] (1,1) -- (1,2);
                \draw[ultra thick, blue, ->] (1,2) -- (2,2);
                \fill[black]
                    (0,0) circle (2pt) node[below] {$f(0,0)$}
                    (2,2) circle (2pt) node[above] {$f(2,2)$};
            \end{tikzpicture}&
            \begin{tikzpicture}
                \thegrid
                \draw[ultra thick, blue, ->] (0,0) -- (0,1);
                \draw[ultra thick, blue, ->] (0,1) -- (1,1);
                \draw[ultra thick, blue, ->] (1,1) -- (2,1);
                \draw[ultra thick, blue, ->] (2,1) -- (2,2);
                \fill[black]
                    (0,0) circle (2pt) node[below] {$f(0,0)$}
                    (2,2) circle (2pt) node[above] {$f(2,2)$};
            \end{tikzpicture}\\
            \begin{tikzpicture}
                \thegrid
                \draw[ultra thick, blue, ->] (0,0) -- (1,0);
                \draw[ultra thick, blue, ->] (1,0) -- (1,1);
                \draw[ultra thick, blue, ->] (1,1) -- (1,2);
                \draw[ultra thick, blue, ->] (1,2) -- (2,2);
                \fill[black]
                    (0,0) circle (2pt) node[below] {$f(0,0)$}
                    (2,2) circle (2pt) node[above] {$f(2,2)$};
            \end{tikzpicture}&
            \begin{tikzpicture}
                \thegrid
                \draw[ultra thick, blue, ->] (0,0) -- (1,0);
                \draw[ultra thick, blue, ->] (1,0) -- (1,1);
                \draw[ultra thick, blue, ->] (1,1) -- (2,1);
                \draw[ultra thick, blue, ->] (2,1) -- (2,2);
                \fill[black]
                    (0,0) circle (2pt) node[below] {$f(0,0)$}
                    (2,2) circle (2pt) node[above] {$f(2,2)$};
            \end{tikzpicture}&
            \begin{tikzpicture}
                \thegrid
                \draw[ultra thick, blue, ->] (0,0) -- (1,0);
                \draw[ultra thick, blue, ->] (1,0) -- (2,0);
                \draw[ultra thick, blue, ->] (2,0) -- (2,1);
                \draw[ultra thick, blue, ->] (2,1) -- (2,2);[]
                \fill[black]
                    (0,0) circle (2pt) node[below] {$f(0,0)$}
                    (2,2) circle (2pt) node[above] {$f(2,2)$};
            \end{tikzpicture}
        \end{tabular}
    \end{center}
    \caption{The six constraints on the computation of $f(2, 2)$.
    }
    \label{fig:polyrec constraints}
\end{figure}
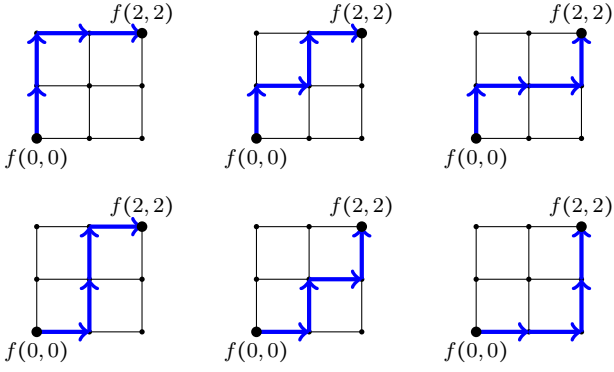

\paragraph{Multivariate \CDA~series}
In the second problem, we consider the notion of
\emph{constructible differentially algebraic} (\CDA) series in commuting variables~\cite{BergeronReutenauer:EJC:1990}.
They were introduced in the multivariate context by Bergeron and Sattler~\cite{BergeronSattler:TCS:1995},
as solutions of systems of \emph{polynomial differential equations} of a restricted form (\cf~\cref{sec:CDA algebra} for details).
Such systems may not have solutions
and deciding whether this is the case is open~\cite[Remark 27]{Clemente:CONCUR:2024}.
Solvability of polynomial differential equations is undecidable
by a result of Denef and Lipshitz~\cite[Theorem 4.11]{DenefLipshitz:1984}
(\cref{rem:undecidability of solvability}),
and thus it is not clear whether \CDA~series have a decidable syntax.
We overcome this issue by showing that commutativity of shuffle automata
can be used to decide whether \CDA~systems are solvable~(\cref{thm:decidability of CDA solvability}),
thus answering the open question positively.

\subsubsection*{Organisation}
In the next section~\cref{sec:preliminaries} we present some necessary preliminaries,
followed in~\cref{sec:commutativity} by a proof of~\cref{thm:commutativity for effective prevarieties}.
Polynomial automata will be studied in~\cref{sec:hadamard},
followed by shuffle automata in~\cref{sec:shuffle}.
In \cref{sec:conclusions} we discuss infiltration automata and other extensions.
More details about infiltration automata, and full proofs, can be found in the appendix.

\section{Preliminaries}
\label{sec:preliminaries}

%
%
We say that a structure (such as a class of series, a class of languages, an algebra, etc...)
is \emph{effective} if its elements can be finitely represented (with a decidable syntax),
equality between elements can be decided based on their representations,
and all the structure operations can be carried over algorithmically.
For instance the class of regular languages is effective.
Most results in the paper hold for any field,
however for computability considerations we restrict our presentation to $\Q$,
which is an effective field.
%
We use $f, g, h$ to denote series in $\series \Q \Sigma$,
and we write a series $f$ as $\sum_{w \in \Sigma^*} f_w \cdot w$,
where the value of $f$ at $w$ is $f_w \in \Q$.
We also write $\coefficient w f$ for $f_w$.
The \emph{support} of a series $f$
is the set of words $w$ such that $f_w \neq 0$.
A \emph{polynomial} is a series with finite support. 
Thus, $3aab - \frac 5 2 bc$ is a polynomial
and $1+a+a^2 + \cdots$ is a series.
The set of series is equipped with a variety of operations.
It carries the structure of a vector space,
with \emph{zero} $\zero$,
\emph{scalar multiplication} $c \cdot f$ with $c \in \Q$,
and \emph{addition} $f + g$ defined element-wise
by $\coefficient w \zero := 0$,
$\coefficient w {(c \cdot f)} := c \cdot \coefficient w f$,
and $\coefficient w {(f + g)} := \coefficient w f + \coefficient w g$,
for every $w \in \Sigma^*$.
This vector space is equipped with two important families of linear maps,
called \emph{left} and \emph{right derivatives}
$\deriveleft a, \deriveright a : \series \Q \Sigma \to \series \Q \Sigma$,
for every $a \in \Sigma$.
They are the series analogues of language quotients:
For every series $f \in \series \Q \Sigma$,
they are defined by $\coefficient w {\deriveleft a f} := \coefficient {a \cdot w} f$,
resp., $\coefficient w {\deriveright a f} := \coefficient {w \cdot a} f$,
for every $w \in \Sigma^*$.
The following elementary observation will be very useful.
\begin{restatable}[Left and right derivatives commute]{lemma}{leftRightComm}
    \label{lem:derive left right commutativity}
    For every two letters $a, b \in \Sigma$,
    we have $\deriveleft a \circ \deriveright b = \deriveright b \circ \deriveleft a$.
\end{restatable}
\noindent
(By $\deriveleft a \circ \deriveright b$ we mean the \emph{composition} of the two derivatives,
which is the function mapping a series $f$ to the series $\deriveleft a \deriveright b f$.
We avoid writing parenthesis for function application when possible.)
We will later equip series with several multiplication operations,
turning them into algebras.
%




%
Let $\SS$ be a class of series,
and for an alphabet $\Sigma$ let $\SS_\Sigma$ the subset of $\SS$
consisting of series over input alphabet $\Sigma$. 
A class of series $\SS$ is a \emph{prevariety} 
if the following two conditions hold, for every alphabet $\Sigma$:
\begin{enumerate}[label=\textbf{(V.\arabic*)}]
    
    \item \label{prevariety:A}
    $\SS_\Sigma$ is a vector space over $\Q$.

    \item \label{prevariety:B}
    For every series $f \in \SS_\Sigma$ and letter $a \in \Sigma$,
    the series $\deriveleft a f$ and $\deriveright a f$ are in $\SS_\Sigma$.

\end{enumerate}
(A \emph{variety}, as introduced by Reutenauer,
is a prevariety which satisfies an additional condition (\cf~\cite[Sec.~III.1]{Reutenauer_1980});
we will not need this notion in the rest of the paper.)
%
A class of series $\SS$ is an \emph{effective prevariety} if it is a prevariety,
and moreover
\begin{enumerate}
    \item every series in $\SS$ admits a finite presentation,
    \item the closure properties \cref{prevariety:A} and \cref{prevariety:B} 
    can be carried over algorithmically by manipulating finite presentations, and
    \item the equality problem 
    is decidable for series in $\SS$.
\end{enumerate}
%
%
\noindent
(By a \emph{finite presentation} for series we mean a concrete syntax
in the same sense, \eg, finite automata are finite presentations for regular languages.)
Thanks to the equivalence $f = g \iff f - g = \zero$ and~\cref{prevariety:A},
the notion of effective prevariety does not change
if we replace the third condition above with
\begin{enumerate}
    \item[3')] the zeroness problem is decidable for series in $\SS$.
\end{enumerate}
\noindent
For instance, the class of rational series is an effective prevariety~\cite[Sec.~III.2.b]{Reutenauer_1980}:
A rational series is finitely represented by a finite weighted automaton,
closure under the vector space operations and both derivatives can be carried over by algorithms manipulating automata,
and the equality problem is decidable~\cite{Schutzenberger:IC:1961}.
\section{Commutativity problem}
\label{sec:commutativity}

We introduce the main problem that we study.
The \emph{commutative image} of a word $w \in \Sigma^*$ (sometimes called \emph{Parikh image})
over a totally ordered alphabet $\Sigma = \set{a_1, \dots, a_d}$
is the vector $\Parikh w := \tuple{\multiplicity {a_1} w, \dots, \multiplicity {a_d} w} \in \N^d$,
where $\multiplicity {a_j} w$ is the number of occurrences of $a_j$ in $w$.
For two words $u, v$ we write $u \sim v$ if they have the same commutative image,
\ie, one word can be obtained from the other by permuting the positions of letters.
For instance, $a_1a_1a_2 \sim a_1a_2a_1 \sim a_2a_1a_1$,
since the commutative image of all three words is $\tuple{2, 1} \in \N^2$,
but $a_1 \not\sim a_2$.
A series $f \in \series \Q \Sigma$ is \emph{commutative} 
(called \emph{échangeable} in~\cite{Fliess:1981})
if for every two words $u, v \in \Sigma^*$ with the same commutative image $u \sim v$,
we have $f_u = f_v$.
Let $\commseries \Q \Sigma$ be the set of commutative series.
It is easy to check that all operations defined in this paper preserve commutativity.



On the face of it, commutativity demands to check infinitely many equalities.
Our simple, but crucial observation is that it can be characterised by finitely many equations.
Clearly, if $f$ is commutative, then for every input symbols $a, b \in \Sigma$ we have
\begin{align}
    \label{eq:swap}
    \tag{\textsf{swap}}
    \deriveleft a \deriveleft b f &= \deriveleft b \deriveleft a f, \quad \text{ and } \\
    \label{eq:rotate}
    \tag{\textsf{rotate}}
    \deriveleft a f &= \deriveright a f.
\end{align}
Indeed, the first equation corresponds to exchanging the first two input letters $abw \sim baw$,
and the second equation corresponds to a rotation $aw \sim wa$.
Since swaps and rotations suffice to generate all commutatively equivalent words,
commutativity admits the following characterisation.
\begin{lemma}[Characterisation of commutativity]
    \label{lem:commutativity}
    A series $f \in \series \Q \Sigma$ is commutative if, and only if,
    it satisfies equations~\cref{eq:swap} and \cref{eq:rotate},
    for every $a, b \in \Sigma$.
\end{lemma}
\noindent
Thanks to \cref{lem:commutativity},
commutativity is decidable for effective prevarieties of series.
This simple but fruitful observation
will allow us to decide commutativity for several classes of series
for which the commutativity problem appears to be nontrivial.
\thmCommutativityForPrevarieties*
\begin{proof}
    Thanks to~\cref{lem:commutativity},
    commutativity reduces to finitely many equality tests~\cref{eq:swap} and~\cref{eq:rotate}.
    Since $\SS$ is a prevariety,
    if $f \in \SS$ then $\deriveleft a f, \deriveright a f \in \SS$ as well.
    By effectiveness, we can construct finite representations for these series
    and decide the equality tests.
\end{proof}
%
%
\begin{remark}
    \Cref{thm:commutativity for effective prevarieties} allows us
    to reduce commutativity to finitely many equality tests.
    The observation that commutativity can be characterised by finitely many equations
    has been exploited before~\cite[Prop.~2]{ChenSongWu:CAV:2016},
    albeit in a different context.
\end{remark}
Thanks to the following observation,
commutativity generalises equality in all our cases.
%
%
\begin{restatable}{lemma}{lemShuffleReflectsCommutativity}
    \label{lem:shuffle reflects commutativity}
    For every $f \in \series \Q \Sigma$
    and fresh input symbols $a, b \not \in \Sigma$,
    the series $ab \shuffle f$ is commutative if, and only if, $f = \zero$.
\end{restatable}
\noindent
(The expression ``$ab \shuffle f$'' denotes the \emph{shuffle product} of $ab$ and $f$,
which will be formally introduced in~\cref{sec:shuffle algebra}.)
%
%
In~\cref{sec:hadamard} and~\cref{sec:shuffle}
we leverage~\cref{thm:commutativity for effective prevarieties}
for two expressive classes of series, together with respective applications.
Extensions will be mentioned in the last section~\cref{sec:conclusions}.

\section{Hadamard automata and series}
\label{sec:hadamard}

In this section, we define a weighted model of computation called \emph{Hadamard automata},
strongly related to polynomial automata.
We show that the series they recognise constitute an effective prevariety\-
(\cref{thm:Hadamard prevariety}),
and thus the commutativity problem is decidable for this model\-
(\cref{thm:Hadamard-finite series - decidable commutativity}).
%
%
We conclude the section by an application to a class of multivariate number sequences\-
(\cref{thm:polyrec consistency}).

\subsection{Preliminaries}

\subsubsection{Difference algebra}
\label{sec:difference algebra}

Let $\A = \tuple{A; {+}, {*}}$ be an \emph{algebra} (over the field of rational numbers),
that is a vector space equipped with a bilinear product~\cite[Ch.~3, §1]{Bourbaki:Algebra:I:1974}.
An algebra \emph{endomorphism} is a linear function $F : A \to A$ \st
\begin{align}
    F(\alpha * \beta) = F \alpha * F \beta,
    \quad\text{for all } \alpha, \beta \in A.
\end{align}
A \emph{difference algebra}~\cite{RMCohn:DifferenceAlgebra:1965,Levin:DifferenceAlgebra:2008}
is an algebra $\tuple{A; {+}, {*}, (F_j)_{1\leq j\leq d}}$
together with finitely many endomorphisms $F_1, \dots, F_d : A \to A$.
We do not require the $F_i$'s to pairwise commute.
%
Difference algebras will provide the semantic background for the rest of the section.

As an example, consider the algebra of multivariate polynomials
$\tuple{\poly \Q X; {+}, {\cdot}}$,
where $X = \tuple{X_1, \dots, X_k}$ is a tuple of commuting variables.
When the variable names do not matter,
we sometimes write $\poly \Q k$.
An endomorphism $F$ acts on a polynomial $\alpha \in \poly \Q X$
just by evaluation $F \alpha = \alpha(F(X_1), \dots, F(X_k))$,
and thus it is uniquely defined once we have fixed $F(X_1), \dots, F(X_k) \in \poly \Q X$.
%
%
If we fix distinguished endomorphisms $F_1, \dots, F_d$,
then the algebra of polynomials acquires the structure of a difference algebra
$\tuple{\poly \Q X; {+}, {\cdot}, (F_j)_{1 \leq j \leq d}}$.
Other examples of difference algebras will be presented in the next sections
\cref{sec:Hadamard algebra} and \cref{sec:shift algebra}.

\subsubsection{Hadamard difference algebra}
\label{sec:Hadamard algebra}

We define a product operation on series
which will be central for the understanding of Hadamard automata.
The \emph{Hadamard product}~\cite{Fliess:1974} of two series $f, g \in \series \Q \Sigma$,
denoted by $f \hadamard g$,
is defined element-wise:
For every $w \in \Sigma^*$,
$\coefficient w (f \hadamard g) := \coefficient w f \cdot \coefficient w g$.
For instance, if we have two polynomials $2aab - c$ and $3aab + b$ (in noncommuting variables),
their Hadamard product is $6aab$.
This is an associative and commutative operation, it distributes over sum,
and the identity element is the series $\one = \sum_{w\in \Sigma} 1 \cdot w$ mapping every word to $1$.
The Hadamard product for series is reminiscent of the notion of \emph{intersection} in language theory
and of \emph{fully synchronous composition} in concurrency theory.

In order to make it easier to compare the Hadamard product
to the shuffle and infiltration products studied later,
we find it convenient to provide the following alternative,
coinductive definition
(\cf~\cite[Definition 8.1]{BasoldHansenPinRutten:MSCS:2017}):
For any two series $f, g \in \series \Q \Sigma$,
their Hadamard product $f \hadamard g$
is the unique series \st
\begin{align}
    \tag{${\hadamard}$-$\e$}
    \label{eq:hadamard:base}
    \coefficient \e {(f \hadamard g)}
        &= f_\e \cdot g_\e, \\
    \tag{${\hadamard}$-$\deriveleft a$}
    \label{eq:hadamard:step}
        \deriveleft a (f \hadamard g)
        &= \deriveleft a f \hadamard \deriveleft a g,
        \ \forall a \in \Sigma.
\end{align}
%
%
Equation \cref{eq:hadamard:step} states that left derivatives are endomorphisms \wrt~the Hadamard product.
We thus obtain the \emph{Hadamard difference algebra} of series
    $\series \Q \Sigma_{\hadamard} := \tuple{\series \Q \Sigma; {+}, {\hadamard}, (\deriveleft a)_{a \in \Sigma}}$,
which will provide the semantic background
for the forthcoming Hadamard automata.
By dropping the difference structure,
we obtain the \emph{Hadamard algebra} $\tuple{\series \Q \Sigma; {+}, {\hadamard}}$.
At this point, it is worth noting that right derivatives are also Hadamard algebra endomorphisms.
This will be useful in the proof of~\cref{lem:Hadamard closure properties}.
\begin{restatable}{lemma}{lemRightDerivationHadamardEndo}
    \label{lem:right derivation Hadamard endomorphism}
    Right derivatives $\deriveright a$ (with $a \in \Sigma$) are Hadamard algebra endomorphisms.
    I.e., they are linear and commute with Hadamard product,
    \begin{align}
        \label{eq:derive hadamard right}
        \deriveright a (f \hadamard g) = \deriveright a f \hadamard \deriveright a g,
        \quad\forall a \in \Sigma.
    \end{align}
\end{restatable}

\subsection{Hadamard automata}
\label{sec:Hadamard automata}

We can now define the central computation model of the section.
A \emph{Hadamard automaton} is a tuple
$A = \tuple{\Sigma, X, F, \Delta}$
where $\Sigma$ is a finite \emph{input alphabet},
$X = \set{X_1,  \dots X_k}$ is a finite set of \emph{nonterminal symbols},
$F : X \to \Q$ is the \emph{output function},
and $\Delta : \Sigma \to X \to \poly \Q X$ is the \emph{transition function}.
A \emph{configuration} of a Hadamard automaton is a polynomial $\alpha \in \poly \Q X$.
For every input symbol $a \in \Sigma$,
the transition function $\Delta_a : X \to \poly \Q X$
is extended to a unique endomorphism $\Delta_a : \poly \Q X \to \poly \Q X$ of the polynomial algebra of configurations.
In other words, $\Delta_a$ acts by evaluation:
$\Delta_a(\alpha) = \alpha(\Delta_a(X_1), \dots, \Delta_a(X_k))$.
Finally, transitions from single letters
are extended to all finite input words homomorphically:
For every configuration $\alpha \in \poly \Q X$,
input word $w \in \Sigma^*$, and letter $a \in \Sigma$,
we have $\Delta_\e \alpha := \alpha$
and $\Delta_{a \cdot w} \alpha := \Delta_w \Delta_a \alpha$.
We have all ingredients to define the \emph{semantics} of a configuration $\alpha \in \poly \Q X$,
which is the series $\sem \alpha \in \series \Q \Sigma$ \st:
\begin{align}
    \label{eq:semantics}
    \sem \alpha_w &:= F \Delta_w \alpha,
    \quad \text{for all } w \in \Sigma^*.
\end{align}
Here, $F$ is extended from nonterminals to configurations also endomorphically,
by $F \beta = \beta(F X_1, \dots, F X_k)$. 
The semantics of a Hadamard automaton $A$
is the series recognised by nonterminal $X_1$.

\begin{remark}
    When the transition function is \emph{linear}
    (i.e., $\Delta_a X_i$ is a polynomial of degree $\leq 1$ for all $a \in \Delta$ and $1 \leq i \leq k$),
    we recover Schützenberger's \emph{weighted finite automata}~\cite{Schutzenberger:IC:1961}.
    Thus Hadamard automata generalise weighted automata.
\end{remark}

\begin{example}
    \label{ex:Hadamard automaton 1}
    Revisiting~\cref{ex:intro:1} in the syntax of Hadamard automata,
    we have a single nonterminal $X = \set{A}$,
    output function $F_c$ defined by $F_c A := c$ ($c \in \Q$) and transitions by
    \begin{align*}
        \Delta_{a_1} A := A^2
            \quad\text{and}\quad
                \Delta_{a_2} A := 1 - A^2.
    \end{align*}
    For instance, when $c = 0$ starting from the initial configuration $A$ and reading input $u = a_1 a_2$
    we obtain configuration $\Delta_{a_2} \Delta_{a_1} A = \Delta_{a_2} A^2 = (1 - A^2)^2$
    and thus $\sem A_u = 1$.
    Reading $v = a_2 a_1$ yields a different configuration
    $\Delta_{a_1} \Delta_{a_2} A = 1 - A^4$,
    however $\sem A_v = 1$.
    In fact, $\sem A_w$ is the count modulo $2$ of the number of $a_2$'s in $w$,
    and thus $\sem A$ is a commutative series.

    It can be checked that also the output functions $F_1, F_{-1}$
    defined by $F_1 A := 1$, resp., $F_{-1} A := -1$
    give rise to commutative series.
    However, no other choice does.
    For instance, $F_2$ defined by $F_2 A := 2$ does not yield a commutative series,
    since $\sem A_{a_1 a_2} = (1 - 2^2)^2 = 9$,
    but $\sem A_{a_2 a_1} = 1 - 2^4 = -15$.
\end{example}

A Hadamard automaton endows the algebra of configurations
with the structure of a difference algebra
$\tuple{\poly \Q X; +, \cdot, (\Delta_a)_{a \in \Sigma}}$.
This allows us to connect the difference algebra of configurations
with the difference algebra of Hadamard series they recognise.
%
\begin{restatable}[Properties of the semantics]{lemma}{lemPropertiesofSemantics}
    \label{lem:Hadamard automata - properties of the semantics}
    The semantics of a Hadamard automaton is a \emph{homomorphism} from
    the difference algebra of polynomials
        $\tuple{\poly \Q X; +, \cdot, (\Delta_a)_{a \in \Sigma}}$
    to the difference Hadamard algebra of series 
        $\tuple{\series \Q \Sigma; +, \hadamard, (\deriveleft a)_{a \in \Sigma}}$.
    In other words, $\sem 0 = \zero$, $\sem 1 = \one$,
    and, for all $\alpha, \beta \in \poly \Q X$,
    \begin{align}
        \label{eq:hadamard:sem:scalar-product}
        \sem {c \cdot \alpha}
            &= c \cdot \sem \alpha 
            && \forall c \in \Q, \\
        \label{eq:hadamard:sem:sum}
            \sem {\alpha + \beta}
            &= \sem \alpha + \sem \beta, \\
        \label{eq:hadamard:sem:product}
            \sem {\alpha \cdot \beta}
            &= \sem \alpha \hadamard \sem \beta, \\
        \label{eq:hadamard:sem:derivation}
        \sem {\Delta_a \alpha}
            &= \deriveleft a \sem \alpha
            && \forall a \in \Sigma.
    \end{align}
\end{restatable}

\begin{remark}[Hadamard vs.~polynomial automata]
    \label{rem:Hadamard and polynomial automata}
    Hadamard automata are very similar to~\emph{polynomial automata}~\cite{BenediktDuffSharadWorrell:PolyAut:2017},
    a common generalisation of weighted automata~\cite{Schutzenberger:IC:1961}
    and vector addition systems~\cite{Mayr:STOC:1981}.
    %
    %
    %
    %
    The two models are equivalent, up to reversal of the input:
    A series $f$ is recognised by a Hadamard automaton
    iff its reversal $f^R$ is recognised by a polynomial automaton.
    (The \emph{reversal} of $f$ is the series $f^R$ mapping $a_1 \cdots a_n$ to $f(a_n \cdots a_1)$.)
    Since commutativity is invariant under reversal,
    our results also apply to polynomial automata.
    %

    Over more general commutative semirings,
    polynomial automata have been introduced as
    \emph{alternating weighted automata}~\cite{KostolanyiMisun:TCS:2018},
    further studied in~\cite{Grabolle:EPTCS:2021}
    (albeit without algorithmic results).
    In the special case of an input alphabet with just one letter $\Sigma = \set a$,
    Hadamard automata have been considered in~\cite{BorealeGorla:CONCUR:2021}
    in a coalgebraic setting, where equality is shown to be decidable\-
    \cite[Theorem 4.1]{BorealeGorla:CONCUR:2021} (without complexity analysis).
\end{remark}

\subsection{Hadamard-finite series}

Hadamard automata provide an operational way to syntactically describe a class of series.
In this section, we explore a semantic way,
which will give us a very short proof of the effective prevariety property
(\cf~\cref{thm:Hadamard prevariety}).

For series $f_1, \dots, f_k \in \series \Q \Sigma$,
let $\poly \Q {f_1, \dots, f_k}_{\hadamard}$ be the smallest Hadamard algebra containing $f_1, \dots, f_k$.
Hadamard algebras of this form are called~\emph{finitely generated},
with \emph{generators} $f_1, \dots, f_k$.
It is a \emph{difference Hadamard algebra}
if it is additionally closed under the endomorphisms $\deriveleft a$, $a \in \Sigma$.
\begin{definition}[Hadamard-finite series]
    \label{def:Hadamard finite}
    A series is \emph{Hadamard finite} if it belongs to a
    finitely generated difference Hadamard algebra.
\end{definition}
The following lemma will be our working definition of Hadamard-finite series.
\begin{restatable}{lemma}{lemHadamardFiniteWorkingDefinition}
    \label{lem:Hadamard finite - working definition}
    A series $f$ is Hadamard finite iff there are generators $f_1, \dots, f_k$ \st:
    \begin{enumerate}
        \item $f \in \poly \Q {f_1, \dots, f_k}_{\hadamard}$.

        \item $\deriveleft a f_i \in \poly \Q {f_1, \dots, f_k}_{\hadamard}$
        for every $a \in \Sigma$ and $1 \leq i \leq k$.
        \Ie, there are polynomials $p^{(a)}_i \in \poly \Q k$ \st
        \begin{align}
            \label{eq:Hadamard finite equations}
            \deriveleft a f_i = p^{(a)}_i(f_1, \dots, f_k),
            \quad\text{for all } a \in \Sigma, 1 \leq i \leq k.
        \end{align}
    \end{enumerate}
    Moreover, we can assume without loss of generality that $f$ equals one of the generators $f = f_1$.
\end{restatable}
\noindent
We call~\cref{eq:Hadamard finite equations}
\emph{systems of Hadamard-finite equations}.
The following lemma states that Hadamard automata recognise precisely the Hadamard-finite series.
Its proof relies on~\cref{lem:Hadamard automata - properties of the semantics,lem:Hadamard finite - working definition}.
\begin{restatable}{lemma}{lemHadamardAutomataAndSeries}
    \label{lem:Hadamard automata and series}
    A series is recognised by a Hadamard automaton if, and only if,
    it is Hadamard finite.
\end{restatable}

\begin{remark}
    Hadamard-finite equations have been considered under the name of
    \emph{polynomial recurrent relations} by Sénizergues\-
    \cite[Definition 5]{Senizergues:CSR:2007},
    where they are claimed, without proof, to coincide
    with those recognised by deterministic pushdown automata of level 3\-
    \cite[Corollary 3]{Senizergues:CSR:2007},
    and to have a decidable equality problem\-
    \cite[Theorem 5]{Senizergues:CSR:2007}.
    In the context of arbitrary semirings,
    they also appear under the name of \emph{Hadamard-polynomial equations}~\cite[Sec.~6]{KostolanyiMisun:TCS:2018}.
\end{remark}

\begin{example}
    \label{ex:Hadamard finite equations 1}
    We illustrate~\cref{lem:Hadamard automata and series}
    on the Hadamard automaton from~\cref{ex:Hadamard automaton 1}.
    The corresponding Hadamard-finite equations are
    %
    $\deriveleft {a_1} \sem A = \sem A^2$ and
    $\deriveleft {a_2} \sem A = \one - \sem A^2$.
    %
    (The square operation is to be understood as Hadamard square,
    \ie, $\sem A^2 = \sem A \hadamard \sem A$;
    recall that $\one = \sum_{w \in \Sigma^*} 1 \cdot w$ is the Hadamard identity.)
    It follows that $\poly \Q {\sem A}_{\hadamard}$
    is a difference Hadamard algebra, and thus $\sem A$ is Hadamard finite.
\end{example}

\subsection{Closure properties of Hadamard-finite series}

We recall some basic closure properties of Hadamard-finite series.
All of them follow directly from the definitions,
except closure under right derivative, which is a novel result.

\begin{restatable}{lemma}{lemHadamardClosureProperties}
    \label{lem:Hadamard closure properties}
    The class of Hadamard-finite series is an effective difference algebra.  
    Moreover, if $f$ is a Hadamard finite,
    then $\deriveright a f$ is also effectively Hadamard finite.
\end{restatable}
\noindent
When we say that $\deriveright a f$ is ``effectively Hadamard finite''
we mean that there is an algorithm that, given in input a finite presentation for $f$ and an input symbol $a \in \Sigma$,
produces in output a finite presentation for $\deriveright a f$.
\begin{proof}
    We have to show that Hadamard-finite series contain $\zero, \one$ (which is obvious),
    and are effectively closed under
    scalar product $c \cdot f$ (with $c \in \Q$), sum $f + g$,
    Hadamard product $f \hadamard g$, left and right derivatives $\deriveleft a f, \deriveright a f$
    (with $a \in \Sigma$).
    Simple manipulations of the generators suffice in each case.
    We show closure under right derivative,
    since it is the only property that is not immediate from the definitions.
    %
    %
    Let $f_1, \dots, f_k$ be generators for $f$ and fix $a \in \Sigma$.
    We consider new generators $\deriveright a f_i$, for every $1 \leq i \leq k$,
    yielding the finitely generated Hadamard algebra 
    %
        $A := \poly \Q {\deriveright a f_1, \dots, \deriveright a f_k}_{\hadamard}$.
    %
    By~\cref{lem:Hadamard finite - working definition},
    the proof is concluded by the following two claims.
    \begin{claim}
        $\deriveright a f \in A$.
    \end{claim}
    \begin{claimproof}
        Since $f$ is in the Hadamard algebra generated by the $f_i$'s
        and right derivative is an endomorphism by~\cref{lem:right derivation Hadamard endomorphism},
        $\deriveright a f$ is in the Hadamard algebra generated by the $\deriveright a f_i$'s,
        \ie, $\deriveright a f \in A$.
    \end{claimproof}
    \begin{claim}
        For every $b \in \Sigma$ and $1 \leq i \leq k$,
        $\deriveleft b \deriveright a f_i \in A$.
    \end{claim}
    \begin{claimproof}
        Since left and right derivatives commute by~\cref{lem:derive left right commutativity},
        we have $\deriveleft b \deriveright a f_i = \deriveright a \deriveleft b f_i$.
        But $\deriveleft b f_i$ is in the Hadamard algebra generated by the $f_i$'s
        and since right derivative is an endomorphism by~\cref{lem:right derivation Hadamard endomorphism},
        $\deriveright a \deriveleft b f_i \in A$.
    \end{claimproof}
\end{proof}

\begin{theorem}
    \label{thm:Hadamard prevariety}
    The class of Hadamard-finite series is an effective prevariety.
\end{theorem}
\begin{proof}
    The two closure conditions~\cref{prevariety:A} and~\cref{prevariety:B}
    follow from~\cref{lem:Hadamard closure properties}.
    By~\cref{lem:Hadamard automata and series} and~\cref{rem:Hadamard and polynomial automata},
    the equality problem for Hadamard-finite series
    is (effectively) equivalent to the equality problem for polynomial automata,
    and the latter is decidable by~\cite[Corollary 1]{BenediktDuffSharadWorrell:PolyAut:2017}.
\end{proof}

We conclude this section by showing that Hadamard-finite series over disjoint alphabets are closed under shuffle product.
This is a surprising observation,
since it is mixing together products of a very different nature.
It will be used in~\cref{thm:Hadamard-finite series - decidable commutativity}
to provide a hardness result for the commutativity problem.
\begin{restatable}{lemma}{lemHadamardShuffle}
    \label{lem:Hadamard shuffle}
    Let $\Sigma, \Gamma$ be two finite and disjoint alphabets $\Sigma \cap \Gamma = \emptyset$.
    If $f \in \series \Q \Sigma$ and $g \in \series \Q \Gamma$ are Hadamard finite,
    then $f \shuffle g$ is effectively Hadamard finite.
\end{restatable}

\subsection{Commutativity problem for Hadamard-finite series}

We now have all ingredients to prove our main result for Hadamard-finite series.
\begin{theorem}
    \label{thm:Hadamard-finite series - decidable commutativity}
    The commutativity problem for Hadamard-finite series (equivalently, Hadamard and polynomial automata) is Ackermann-complete.
\end{theorem}
\begin{proof}
    Thanks to~\cref{thm:commutativity for effective prevarieties} and \cref{thm:Hadamard prevariety},
    the commutativity problem for Hadamard-finite series is decidable.
    In fact, we have reduced the commutativity problem over alphabet $\Sigma$
    to $\card \Sigma^2 + \card \Sigma$ equivalence queries (\cf~\cref{lem:commutativity}).
    
    On the other hand, the equivalence problem for Hadamard-finite series
    efficiently reduces to the commutativity problem (\ie, in polynomial time).
    Indeed, let $f \in \series \Q \Sigma$ be a Hadamard-finite series
    and consider two fresh input symbols $a, b \not\in \Sigma$.
    Thanks to~\cref{lem:shuffle reflects commutativity},
    $g := ab \shuffle f$ is commutative if, and only if, $f = \zero$.
    Moreover, since $ab$ is Hadamard-finite (even rational),
    $g$ is Hadamard finite by~\cref{lem:Hadamard shuffle}.
    Furthermore, a finite representation for $g$ can be constructed in polynomial time from a finite representation of $f$.
    Thus, we have reduced checking whether $f = \zero$
    to checking whether $g$ is commutative.

    Summarising, commutativity and equivalence are polynomial-time equivalent.
    The complexity of equivalence for polynomial automata is Ackermann-complete~\cite[Theorem 1 + Corollary 1]{BenediktDuffSharadWorrell:PolyAut:2017},
    thus the same is true for Hadamard-finite series~(\cf~\cref{rem:Hadamard and polynomial automata}).
    %
\end{proof}

\subsection{Application: Multivariate polynomial recursive sequences}
\label{sec:multivariate polyrec}

We introduce a class of multivariate recursive sequences
generalising the univariate \emph{polynomial recursive sequences} (in short, \emph{polyrec})
from~\cite{Cadilhac:Mazowiecki:Paperman:Pilipczuk:Senizergues:ToCS:2021}.
We will make a crucial use of~\cref{thm:Hadamard-finite series - decidable commutativity}
to argue that this class can be presented \emph{effectively} by polynomial recursive equations.
In other words, we will show that polynomial recursive equations are a \emph{decidable syntax} for polyrec sequences:
There is an algorithm which, given a system of polynomial recursive equations in input,
decides whether they actually define a (polyrec) sequence at all.
The lack of an effective presentation
has been an obstacle to the study of a multivariate generalisation of polyrec sequences.
In~\cref{sec:shift algebra} we begin with some preliminaries about multivariate sequences,
in~\cref{sec:polyrec} we define multivariate polyrec sequences and study some of their basic properties,
and finally in~\cref{sec:consistency of polyrec equations} we use~\cref{thm:Hadamard-finite series - decidable commutativity}
to argue that their syntax is effective.

\subsubsection{Difference sequence algebra}
\label{sec:shift algebra}

Fix a dimension $d \in \N_{\geq 1}$.
For every $1 \leq j \leq d$,
by $e_j$ we denote the $j$-th unit vector in $\N^d$.
A \emph{multivariate sequence} is a function $f : \N^d \to \Q$.
We equip the set of sequences with the structure of a difference algebra
$\tuple{\N^d \to \Q; +, \cdot, (\shift j)_{1 \leq j \leq d}}$,
where $\zero$, $\one$, scalar product $c \cdot f$ ($c \in \Q$),
sum $f + g$, and multiplication $f \cdot g$ are defined element-wise,
and the \emph{$j$-th left shift} endomorphism $\shift j$ maps a sequence $f : \N^d \to \Q$
to the sequence $\shift j f$ \st
\begin{align*}
    (\shift j f)(n) := f(n + e_j),
    \quad \text{for all } n \in \N^d.
\end{align*}
Unlike the difference algebra of Hadamard series from~\cref{sec:Hadamard algebra},
the endomorphisms $\shift j$'s pairwise commute
$\shift j \circ \shift h = \shift h \circ \shift j$.

\subsubsection{Polyrec sequences}
\label{sec:polyrec}

A $d$-variate sequence $f : \N^d \to \Q$ is \emph{polyrec}
if there exist $k \in \N_{\geq 1}$,
auxiliary sequences $f_1, \dots, f_k : \N^d \to \Q$ with $f = f_1$,
and polynomials $p^{(j)}_i \in \poly \Q k$ for all $1 \leq i \leq k$ and $1 \leq j \leq d$, \st
\begin{align}
    \label{eq:polyrec}
    \shift j f_i
        = p^{(j)}_i(f_1, \dots, f_k),
        \quad \forall 1 \leq i \leq k, 1 \leq j \leq d.
\end{align}
%
%
In other words, for each $f_i$ and coordinate $1 \leq j \leq d$
there is a distinct recursive equation
specifying how future values can be computed from the current one.

The univariate case $d = 1$ corresponds to the polyrec sequences
from~\cite{Cadilhac:Mazowiecki:Paperman:Pilipczuk:Senizergues:ToCS:2021},
a rich class containing the linear recursive sequences (such as the Fibonacci numbers),
as well as fast growing sequences such as $2^{2^n}$.
Moreover, the zeroness and equivalence problems for univariate polyrec sequences are decidable
(with Ackermann complexity~\cite{BenediktDuffSharadWorrell:PolyAut:2017}),
and whether there is an elementary algorithm has been open for some time
(\cf~\cite{ClementeDontenBuryMazowieckiPilipczuk:STACS:23}).
In order to develop intuition, we show below some examples of polyrec sequences.
\begin{example}
    We begin with a simple example, univariate and linear.
    The \emph{Fibonacci sequence} $F : \N \to \Q$
    is well-known to satisfy the linear recursion $F(n+2) = F(n+1) + F(n)$.
    We can turn this into the polyrec format by introducing an auxiliary sequence $G : \N \to \Q$ \st\-
    $\shift 1 F = F + G$ and $\shift 1 G = F$.
    %
    In this way, one shows that all linear recursive sequences are polyrec.
\end{example}

\begin{example}
    \label{ex:polyrec 1}
    We now consider a bivariate example $f : \N^2 \to \Q$.
    For every $n_1, n_2 \in \N$, let 
    $f(n_1, n_2) := 2^{2^{n_1}} \cdot 2^{3^{n_2}}$.
    Introducing auxiliary sequences $g(n_1, n_2) := 2^{2^{n_1}}$ and $h(n_1, n_2) := 2^{3^{n_2}}$
    we have polyrec equations
    \begin{align*}
        \arraycolsep=1pt
        \begin{array}{rl}
            \shift 1 f &= f \cdot g, \\
            \shift 2 f &= f \cdot h^2,
        \end{array}
        \quad
        \begin{array}{rl}
            \shift 1 g &= g^2, \\
            \shift 2 g &= g,
        \end{array}
        \quad
        \begin{array}{rl}
            \shift 1 h &= h, \\
            \shift 2 h &= h^3.
        \end{array}
    \end{align*}
\end{example}

\begin{example}
    The bivariate factorial $f(n_1, n_2) := n_1! \cdot n_2!$ satisfies the recursion
    \begin{align*}
        f(n_1+1, n_2) &= (\underline {n_1} + 1) \cdot f(n_1, n_2), \\
        f(n_1, n_2+1) &= (\underline {n_2} + 1) \cdot f(n_1, n_2).
    \end{align*}
    This is not in the polyrec format~\cref{eq:polyrec},
    because of the two underlined occurrences of the index variables $n_1, n_2$.
    Linear recursions with coefficients polynomial in $n_1, \dots, n_d$
    are called \emph{P-recursive}~\cite{Zeilberger:JMAA:1982}.
    We can introduce auxiliary sequences $g(n_1, n_2) := n_1$ and $h(n_1, n_2) := n_2$
    and transform P-recursions into polyrec equations
    \begin{align*}
        \arraycolsep=1pt
        \begin{array}{rl}
            \shift 1 f &= (g+\one) \cdot f, \\
            \shift 2 f &= (h+\one) \cdot f,
        \end{array}
        \quad
        \begin{array}{rl}
            \shift 1 g &= \one + g, \\
            \shift 2 g &= g,
        \end{array}
        \quad
        \begin{array}{rl}
            \shift 1 h &= h, \\
            \shift 2 h &= \one + h.
        \end{array}
    \end{align*}
    A similar argument shows that $n_1! + n_2!$ is polyrec.
    In fact, polyrec include all $P$-recursive definitions with a constant leading term
    (\cf~\cite[Proposition 3.6]{Lipshitz:D-finite:JA:1989}),
    also called \emph{monic $P$-recursive} in\-
    \cite[page 5]{Buna-MargineanChevalShirmohammadiWorrell:POPL:2024}.
    For instance, $n!$ is monic $P$-recursive since $(n+1)! = (n+1) \cdot n!$,
    but the \emph{Catalan numbers} $C_n$ are not since $\underline {(n+2)} \cdot C_{n+1} = (4n+2) \cdot C_n$
    (the underlined term is non-constant).
\end{example}
\begin{remark}
    All the examples so far are sums of products of univariate polyrec sequences.
    This needs not be the case in general.
    For instance, $(n_1 + n_2)!$ and $2^{2^{n_1 + n_2}}$ are polyrec,
    but not of that form.
    This shows that the characterisation of multivariate \emph{linear} recursive sequences
    in~\cite[Proposition 2.1]{Karhumaki:TCS:1977} as sums of products of univariate linear recursive sequences
    does not generalise to the polyrec case.
\end{remark}
\noindent
More examples will be shown in~\cref{sec:consistency of polyrec equations}.
%
%
Multivariate polyrec sequences are a well-behaved class of sequences.
They are closed under the algebra operations (scalar product, sum, and multiplication),
shifts $\shift j f$,
\emph{sections}~\cite[Definition 3.1]{Lipshitz:D-finite:JA:1989},
and \emph{diagonals}~\cite[Definition 2.6]{Lipshitz:D-finite:JA:1989}.
More details about closure properties can be found in~\cref{app:polyrec}.
%

\subsubsection{Consistency of polyrec equations}
\label{sec:consistency of polyrec equations}

So far we have defined a class of multivariate sequences, and we have argued that this class is robust
by presenting examples and closure properties.
However, we have sidestepped a crucial issue,
namely whether the equations~\cref{eq:polyrec} together with an initial condition
do actually define a sequence.
By the shape of the equations,
a given initial condition gives rise to \emph{at most one} solution.
However, whether a solution exists at all is nontrivial.
\begin{problem}[Polyrec consistency problem]
    \label{problem:polyrec consistency}
    The \emph{consistency problem} for polyrec equations
    takes as input a set of polyrec equations~\cref{eq:polyrec}
    together with an initial condition $c \in \Q^k$,
    and it amounts to decide whether there is a sequence solution $f \in (\N^d \to \Q)^k$ 
    extending the initial condition $f(0) = c$.
\end{problem}
\noindent
We emphasise that~\cref{problem:polyrec consistency} is about polyrec equations,
which may or may not represent a polyrec sequence,
and the question is finding out which one is the case.
We show that~\cref{problem:polyrec consistency} can algorithmically be decided,
thus showing that the syntax of polyrec sequences~\cref{eq:polyrec} is \emph{effective}.
%
%
In the \underline{univariate} case $d = 1$,
\emph{every} initial condition extends to a solution,
and thus the consistency problem is trivial.
However, in the multivariate case $d \geq 2$,
whether a system of equations~\cref{eq:polyrec}
actually defines a (tuple of) sequence(s)
involves an infinite number of constraints to be satisfied\-
(\cf~\cref{fig:polyrec constraints}).
%

\begin{example}[All initial conditions extend to a solution]
    Let $d = 2$ and consider polyrec equations
    \begin{align*}
        \shift 1 f = f^3
            \quad\text{and}\quad
                \shift 2 f = f^5.
    \end{align*}
    The update functions $p(x) := x^3$ and $q(x) := x^5$
    commute $p(q) = q(p) = x^{15}$ identically. 
    This is a very strong condition,
    implying that \emph{every} initial condition $f(0, 0) := c \in \Q$
    extends to a sequence solution,
    namely $f(n_1, n_2) = c^{3^{n_1} \cdot 5^{n_2}}$.
\end{example}

\begin{example}[Some initial conditions extend to a solution]
    \label{ex:polyrec equations}
    Revisiting~\cref{ex:intro:2},
    let $d = 2$ and consider polyrec equations
    \begin{align*}
        \shift 1 f = f^2
            \quad\text{and}\quad
                \shift 2 f = \one - f^2.
    \end{align*}
    We claim that the set of initial values extending to a solution is $\set{-1, 0, 1}$.
    Indeed, if we let $p(x) := x^2$ and $q(x) := 1 - x^2$,
    then the equation $p(q(x)) = q(p(x))$ has exactly three solutions $x \in \set{-1, 0, 1}$.
    %
    This set is stable under $p(x)$ and $q(x)$:
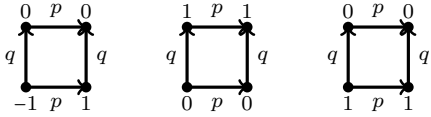
    \begin{center}
    \newcommand{\littlediagram}[4]{
        \tikzset{x=8mm, y=8mm, every node/.append style={font=\footnotesize}}
        \begin{tikzpicture}
            \draw[step=8mm] (0,0) grid (1,1);
            \foreach \i in {0,...,1}
                \foreach \j in {0,...,1}
                    \fill (\i,\j) circle (1pt);
            \draw[very thick, ->] (0,0) -- (1,0) node[midway, below] {$p$};
            \draw[very thick, ->] (0,0) -- (0,1) node[midway, left] {$q$};
            \draw[very thick, ->] (0,1) -- (1,1) node[midway, above] {$p$};
            \draw[very thick, ->] (1,0) -- (1,1) node[midway, right] {$q$};
            \fill[black]
                (0,0) circle (2pt) node[below] {$#1$}
                (1,0) circle (2pt) node[below] {$#2$}
                (0,1) circle (2pt) node[above] {$#3$}
                (1,1) circle (2pt) node[above] {$#4$};
            \end{tikzpicture}
        }
    \begin{tabular}{ccc}
        \littlediagram {-1} 1 0 0&
        \littlediagram 0 0 1 1&
        \littlediagram 1 1 0 0
    \end{tabular}
    \end{center}
    Consequently, any $f(0, 0) \in \set{-1, 0, 1}$ extends to a solution.
\end{example}


\begin{example}[No initial condition extends to a solution]
    Let $d = 2$ and consider the polyrec equations
    \begin{align}
        \label{eq:polyrec example 2}
        \shift 1 f = f^2
            \quad\text{and}\quad
                \shift 2 f = f + \one.
    \end{align}
    Let the initial condition be $f(0, 0) := c \in \Q$.
    From the first equation we have $f(1, 0) = c^2$,
    and from the second $f(0, 1) = c+1$.
    But then we can obtain $f(1, 1)$ in two ways,
    $c^2 + 1$ and $(c+1)^2 = c^2 + 2c + 1$,
    implying $c = 0$ and thus $f(1, 1) = 1$.
    But now we compute $f(2, 2)$ in two ways
    obtaining $1^2 + 1 = 2$ and $(1+1)^2 = 4$.
    This leads to a contradiction,
    and thus there is no initial condition extending to a solution of~\cref{eq:polyrec example 2}.
\end{example}


%
\noindent
\cref{problem:polyrec consistency} is an instance of the more general problem
solvability in sequences for systems of polynomial difference equations,
not necessarily of the polyrec kind~\cref{eq:polyrec}.
We now recall that the latter problem is undecidable in full generality.
\begin{remark}
    \label{rem:undecidability of polynomial difference equations}
    The more general problem of solvability in sequences
    for systems of polynomial difference equations
    is undecidable~\cite[Proposition 3.9]{PogudinScanlonWibmer:2020}.
    Nonlinear bivariate systems of two unknowns $k = d = 2$ suffice.
    The undecidable systems~\cite[Eq.~(5) on pg.~21]{PogudinScanlonWibmer:2020}
    in our notation can be written as follows:
    \begin{align*}
        (f - 1)(f - 2) \cdots (f - N) &= 0, \\
        (g - 1)(g - 2) \cdots (g - N) &= 0, \\
        \prod_{i = 1}^\ell ((a_\ell - f)^2 + (b_\ell - \shift 2 f)^2 + (c_\ell - g)^2 + (d_\ell - \shift 1 g)^2) &= 0,
    \end{align*}
    for suitable constants $N, a_1, \dots, a_\ell, b_1, \dots, b_\ell, c_1, \dots, c_\ell \in \N$,
    where for simplicity, we identify $n \in \N$ with $n \cdot \one$.
    (This is used to encode plane tiling problems,
    where the first two constraints force $f, g : \N^d \to \set {1, \dots, N}$ to range over a finite set,
    and the last equation encodes horizontal and vertical tiling constraints.)
    %
    None of the equations above is in the polyrec format~\cref{eq:polyrec}.
\end{remark}

Consequently, in order to solve~\cref{problem:polyrec consistency}
we need to exploit the special shape of polyrec.
To this end, we develop a connection between sequence solutions of polyrec equations~\cref{eq:polyrec}
and Hadamard-finite series.
%

We make the very simple observation
that number sequences and commutative series contain the same information.
For instance, the sequence $2^{n_1} 3^{n_2}$ is ``the same'' as the series that maps
$w \in \set{a_1, a_2}^*$ to $2^{\multiplicity {a_1} w} 3^{\multiplicity {a_2} w}$.
This is made precise by saying that the difference sequence algebra
$\tuple{\N^d \to \Q; +, \cdot, (\shift j)_{1 \leq j \leq d}}$
and the difference Hadamard algebra of commutative series
$\tuple{\commseries \Q \Sigma; +, \cdot, (\deriveleft {a_j})_{1 \leq j \leq d}}$
(with $\Sigma = \set{a_1, \dots, a_d}$)
are isomorphic.
The isomorphism maps a commutative series $f \in \commseries \Q \Sigma$
to the number sequence $\widetilde f : \N^d \to \Q$
\st~$\widetilde f(n_1, \dots, n_d) := f(a_1^{n_1} \cdots a_d^{n_d})$ for every $n_1, \dots, n_d \in \N$.
This map preserves the algebra operations and the endomorphisms,
which has pleasant consequences.

Given a system of polyrec equations~\cref{eq:polyrec},
we construct a system of Hadamard-finite equations
\begin{align}
    \label{eq:Hadamard-finite companion system}
    \deriveleft {a_j} g_i = p^{(j)}_i(g_1, \dots, g_k),
    \quad \forall 1 \leq i \leq k, 1 \leq j \leq d,
\end{align}
which we call the \emph{companion system} to~\cref{eq:polyrec}.
\begin{example}
    The polyrec equations from~\cref{ex:polyrec equations}
    $\set{\shift 1 f = f^2, \shift 2 f = \one-f^2}$
    yield the companion system from~\cref{ex:Hadamard finite equations 1}
    $\set{\deriveleft {a_1} g = g^2, \deriveleft {a_2} g = \one - g^2}$.
\end{example}
\noindent
Every initial condition extends to a \emph{unique} series solution $g = \tuple{g_1, \dots, g_k} \in \series \Q \Sigma^k$
of~\cref{eq:Hadamard-finite companion system},
which can be proved by induction on the length of words.
By definition, this solution is a tuple of Hadamard-finite series.
When all components $g_1, \dots, g_k$ are commutative,
we can extract corresponding number sequences
$f_1 := \widetilde {g_1}, \dots, f_k := \widetilde {g_k} : \N^d \to \Q$,
which are precisely the unique sequence solution
to the original polyrec equations~\cref{eq:polyrec}:
\begin{align*}
    \shift j {f_i}
    = \widetilde {\deriveleft {a_j} g_i}
    = \widetilde {\left(p^{(j)}_i(g_1, \dots, g_k)\right)}
    = p^{(j)}_i(f_1, \dots, f_k).
\end{align*}
On the other hand, when any component $g_i$ is noncommutative,
\cref{eq:polyrec} does not have any number sequence solution.
To see this, it suffices to notice that any sequence solution of~\cref{eq:polyrec}
gives rise to a commutative Hadamard-finite series solution of the companion system~\cref{eq:Hadamard-finite companion system}
(via the inverse of the isomorphism above),
but as we have observed~\cref{eq:Hadamard-finite companion system} has a unique solution.

The argument above gives us an algorithm for the consistency problem:
Construct the companion system~\cref{eq:Hadamard-finite companion system}
and check that all components of the unique Hadamard-finite series solution are commutative,
which is decidable by~\cref{thm:Hadamard-finite series - decidable commutativity}.
\begin{theorem}
    \label{thm:polyrec consistency}
    The polyrec consistency problem~(\cref{problem:polyrec consistency}) is decidable
    with Ackermann complexity.
\end{theorem}
\section{Shuffle automata and series}
\label{sec:shuffle}

In this section we recall \emph{weighted basic parallel processes}~\cite{Clemente:CONCUR:2024},
a nonlinear weighted model of computation generalising weighted finite automata
inspired from the theory of concurrency.
%
For reasons of uniformity with the rest of the paper,
we will refer to it as \emph{shuffle automata}.
We will briefly recall the necessary basic results on shuffle automata from~\cite{Clemente:CONCUR:2024},
to which we refer for more details.
The novel technical result in this section is that the series recognised by shuffle automata
are effectively closed under right derivatives (\cref{lem:shuffle finite right derivative}),
and thus constitute an effective prevariety (\cref{thm:shuffle finite prevariety}).
The closure property is nontrivial,
since the definition of shuffle automata is asymmetric \wrt~right and left derivatives.
Together with the \TWOEXPSPACE~algorithm for the equivalence problem~\cite[Theorem 1]{Clemente:CONCUR:2024},
we obtain that commutativity is decidable, with the same complexity (\cref{thm:shuffle finite commutativity problem}).
We then present an application to multivariate \emph{constructible differentially algebraic} power series (\CDA)
introduced by Bergeron and Sattler~\cite{BergeronSattler:TCS:1995}.
We show that their syntax is effective\-
(\cref{thm:decidability of CDA solvability}),
which was an open question\- 
\cite[Remark 27]{Clemente:CONCUR:2024}.


\subsection{Preliminaries}

\subsubsection{Differential algebra}

A \emph{derivation} of an algebra $\tuple{A; {+}, {*}}$ is a linear function $\delta : A \to A$
satisfying the following product rule (\emph{Leibniz rule})
\begin{align}
    \label{eq:Leibniz rule}
    \delta(\alpha * \beta) = \delta \alpha * \beta + \alpha * \delta \beta,
    \quad\text{for all } \alpha, \beta \in A.
\end{align}
A \emph{differential algebra}
$\tuple{A; {+}, {*}, (\delta_j)_{1\leq j\leq d}}$ is an algebra
together with finitely many derivations $\delta_1, \dots, \delta_d : A \to A$\-
\cite{Ritt:DA:1950,Kaplansky:DA:1957,Kolchin:1973}.
We do not require the $\delta_i$'s to pairwise commute.
As an example, consider the algebra of polynomials $\tuple{\poly \Q X; {+}, {\cdot}}$.
A derivation $\Delta : \poly \Q X \to \poly \Q X$ acts on a polynomial by linearity and~\cref{eq:Leibniz rule}.
For instance, $\Delta (X_1 \cdot X_2) = \Delta X_1 \cdot X_2 + X_1 \cdot \Delta X_2$
and $\Delta (X_1^5) = 5 \cdot X_1^4 \cdot \Delta X_1$.
Thus, a derivation $\Delta$ of the polynomial algebra is uniquely determined
once we fix $\Delta X_1, \dots, \Delta X_k$.
For instance, the familiar partial derivative $\frac \oldpartial  {\oldpartial X_i}$
(for short, $\partial {X_i} {}$)
is the unique derivation $\Delta$ \st~$\Delta X_i = 1$ and $\Delta X_j = 0$ for $j \neq i$.
If we fix distinguished derivations $\Delta_1, \dots, \Delta_d$,
then we obtain the polynomial differential algebra
$\tuple{\poly \Q X; {+}, {\cdot}, (\Delta_j)_{1 \leq j \leq d}}$.
Polynomial differential algebras will provide the state space of configurations for shuffle automata.
Other examples of differential algebras will be presented in
\cref{sec:shuffle algebra} and \cref{sec:CDA algebra}.

\subsubsection{Shuffle differential algebra}
\label{sec:shuffle algebra}

%
We recall a product operation on series which will be central
to the semantics of shuffle automata.
The \emph{shuffle product} of two series $f \shuffle g$ is defined coinductively by
(\cf~\cite[Definition 8.1]{BasoldHansenPinRutten:MSCS:2017})
\begin{align}
    \tag{${\shuffle}$-$\e$}
    \label{eq:shuffle:base}
    \coefficient \e {(f \shuffle g)}
        &= f_\e \cdot g_\e, \\
    \tag{${\shuffle}$-$\deriveleft a$}
    \label{eq:shuffle:step}
        \deriveleft a (f \shuffle g)
        &= \deriveleft a f \shuffle g + f \shuffle \deriveleft a g,
        \quad \forall a \in \Sigma.
\end{align}
This is an associative and commutative operation,
with identity the series $1 \cdot \e$.
%
It originates in the work of Eilenberg and MacLane in homological algebra~\cite{EilenbergMacLane:AM:1953},
and was introduced in automata theory by Fliess under the name of \emph{Hurwitz product}~\cite{Fliess:1974}.
It is the series analogue of the shuffle product in language theory,
and it finds applications in concurrency theory,
where it models the \emph{interleaving semantics} of process composition.
Intuitively, it models all possible ways in which two sequences can be interleaved.
For instance, when we run processes $ab$ and $a$ in parallel we obtain
\begin{align*}
    &ab \shuffle a
        = a (\deriveleft a (ab \shuffle a)) + b (\deriveleft b (ab \shuffle a)) = \\
        &= a (\deriveleft a {(ab)} \shuffle a + ab \shuffle \deriveleft a a) + b (0)
        = a (b \shuffle a + ab \shuffle \e) = \\
        &= a (ba + ab + ab) = 2aab + aba.
\end{align*}
Thus, there are two executions over input $aab$
and one over $aba$.
By~\cref{eq:shuffle:step}, left derivatives $\deriveleft a$ ($a \in \Sigma$) are derivations for the shuffle product,
and we have thus obtained the \emph{differential shuffle algebra}
$\series \Q \Sigma_{\shuffle} := \tuple{\series \Q \Sigma; {+}, {\shuffle}, (\deriveleft a)_{a \in \Sigma}}$.
This is where the semantics of shuffle automata lives.

Before proceeding, it is worth noting that right derivatives are also derivations of the shuffle algebra.
This will be useful in the proof of~\cref{lem:shuffle finite right derivative}.
\begin{restatable}{lemma}{lemRightDerivationShuffleDer}
    \label{lem:right derivation shuffle derivation}
    Right derivatives $\deriveright a$ (with $a \in \Sigma$) are shuffle algebra derivations.
    I.e., they are linear and satisfy Leibniz rule~\cref{eq:Leibniz rule} for the shuffle product:
    \begin{align}
        \label{eq:derive shuffle right}
        \deriveright a (f \shuffle g)
        = \deriveright a f \shuffle g + f \shuffle \deriveright a g,
        \quad\forall a \in \Sigma.
    \end{align}
\end{restatable}


\subsection{Shuffle automata and shuffle-finite series}
\label{sec:shuffle automata}

Syntactically, a shuffle automaton $A = \tuple{\Sigma, X, F, \Delta}$
is identical to a Hadamard automaton~(\cf~\cref{sec:Hadamard automata}),
however the semantics is different:
The transition $\Delta_a : X \to \poly \Q X$ (for $a \in \Sigma$)
is extended to a unique \emph{derivation} $\Delta_a : \poly \Q X \to \poly \Q X$
of the differential polynomial algebra of configurations (\cf~\cref{eq:Leibniz rule}).
The fact that such an extension exists and is unique is a basic fact from differential algebra, \cf~\cite[page 10, point 4]{Kaplansky:DA:1957}.
The extension to all finite words and the definition of the semantics of a configuration are unchanged
(\cf~\cref{eq:semantics}).
The series recognised by the shuffle automaton $A$ is $\sem {X_1}$.

\begin{example}
    \label{ex:binomial shuffle automaton}
    Consider a binary alphabet $\Sigma = \set{a_1, a_2}$
    and three nonterminals $X = \set{X_1, X_2, X_3}$.
    Let the output function be $F X_1 := 1$ and $F X_2 := F X_3 := 0$,
    and consider transitions
    \begin{align*}
        \arraycolsep=1pt
        \begin{array}{rl}
            \Delta_{a_1} X_1 &:= X_1 \cdot (1 + X_3), \\
            \Delta_{a_2} X_1 &:= X_1 \cdot X_2, \\
        \end{array}
        \quad 
        \begin{array}{rl}
            \Delta_{a_1} X_2 &:= 1, \\
            \Delta_{a_2} X_2 &:= 0, \\
        \end{array}
        \quad 
        \begin{array}{rl}
            \Delta_{a_1} X_3 &:= 0, \\
            \Delta_{a_2} X_3 &:= 1.
        \end{array}
    \end{align*}
    The run starting from $X_1$ and reading input $w = a_1 a_2$ is
    \begin{align*}
        X_1
            &\goesto {a_1} X_1 \cdot (1 + X_3)
            \goesto {a_2} \Delta_{a_2} (X_1 \cdot (1 + X_3)) = \\
            &= \Delta_{a_2} X_1 \cdot (1 + X_3) + X_1 \cdot \Delta_{a_2} (1 + X_3) = \\
            &= X_1 \cdot X_2 \cdot (1 + X_3) + X_1.
    \end{align*}
    and thus $\sem {X_1}_{a_1 a_2} = F (X_1 \cdot X_2 \cdot (1 + X_3) + X_1) = 1$.
    We will see in~\cref{ex:binomial CDA} that the series recognised by $X_1$
    maps $w$ to ${n \choose k} \cdot k!$ where $n := \multiplicity {a_1} w$ and $k := \multiplicity {a_2} w$.
\end{example}

Shuffle automata have been introduced in~\cite{Clemente:CONCUR:2024}
as a simultaneous generalisation of weighted finite automata~\cite{Schutzenberger:IC:1961}
and a model of concurrency called \emph{basic parallel processes}~\cite{Esparza:FI:1997}.
In the rest of the section,
we recall their basic properties.
A shuffle automaton endows the algebra of configurations
with the structure of a differential algebra
$\tuple{\poly \Q X; +, \cdot, (\Delta_a)_{a \in \Sigma}}$.
The following connection with the differential algebra of series
bears similarity to~\cref{lem:Hadamard automata - properties of the semantics}.
\begin{restatable}
    [Properties of the semantics~\protect{\cite[Lemma 8 + Lemma 9]{Clemente:CONCUR:2024}}]
    {lemma}{lemPropertiesofSemanticsOfShuffleAutomata}
    \label{lem:properties of semantics - shuffle automata}
    The semantics of a shuffle automaton is a homomorphism from
    the differential algebra of polynomials
        $\tuple{\poly \Q X; +, \cdot, (\Delta_a)_{a \in \Sigma}}$
    to the differential shuffle algebra of series 
        $\tuple{\series \Q \Sigma; +, \shuffle, (\deriveleft a)_{a \in \Sigma}}$.
    In other words, $\sem 0 = \zero$, $\sem 1 = 1 \cdot \e$,
    and, for all configurations $\alpha, \beta \in \poly \Q X$,
    \begin{align}
        \label{eq:shuffle:sem:scalar-product}
        \sem {c \cdot \alpha}
            &= c \cdot \sem \alpha 
            && \forall c \in \Q, \\
        \label{eq:shuffle:sem:sum}
            \sem {\alpha + \beta}
            &= \sem \alpha + \sem \beta, \\
        \label{eq:shuffle:sem:product}
            \sem {\alpha \cdot \beta}
            &= \sem \alpha \shuffle \sem \beta, \\
        \label{eq:shuffle:sem:derivation}
        \sem {\Delta_a \alpha}
            &= \deriveleft a \sem \alpha
            && \forall a \in \Sigma.
    \end{align}
\end{restatable}
Shuffle automata are data structures representing series.
The same class of series admits a semantic presentation,
which we find more convenient to work with and which we now recall.
For series $f_1, \dots, f_k \in \series \Q \Sigma$,
we denote by $\poly \Q {f_1, \dots, f_k}_{\shuffle}$ the smallest shuffle algebra containing $f_1, \dots, f_k$,
which we call~\emph{finitely generated} with \emph{generators} $f_1, \dots, f_k$.
When it is closed under the derivations $\deriveleft a$ (for all $a \in \Sigma$),
it is called a \emph{differential shuffle algebra}.
A series $f \in \series \Q \Sigma$ is \emph{shuffle finite}~\cite[Section 2.4]{Clemente:CONCUR:2024}
if it belongs to a finitely generated differential shuffle algebra
(note the analogy with Hadamard-finite series, \cf~\cref{def:Hadamard finite}).
By~\cite[Lemma 39]{Clemente:WBPP:arXiv:2024},
this is equivalent to the existence of
generators $f_1, \dots, f_k \in \series \Q \Sigma$ \st:
\begin{enumerate}
    \item $f \in \poly \Q {f_1, \dots, f_k}_{\shuffle}$, and
    \item $\deriveleft a f_i \in \poly \Q {f_1, \dots, f_k}_{\shuffle}$,
    for all $a \in \Sigma$ and $1 \leq i \leq k$.
\end{enumerate}
Thanks to this, one can prove that shuffle-finite series coincide
with the series recognised by shuffle automata~\cite[Theorem 12]{Clemente:CONCUR:2024},
and constitute an effective differential algebra~\cite[Lemma 10]{Clemente:CONCUR:2024}.
We complete the picture by observing that they are effectively closed under right derivatives.
%
\begin{lemma}
    \label{lem:shuffle finite right derivative}
    If $f$ is shuffle finite, then $\deriveright a f$ is also effectively shuffle finite.
\end{lemma}
\noindent
The definition of shuffle-finite series is asymmetric
since it is biased towards left derivatives,
therefore closure under right derivatives is nontrivial.
Nonetheless, the technology of shuffle-finite series
gives a rather short proof of this fact,
showing their effectiveness.
%
Thanks to the lemma and~\cite[Theorem 12]{Clemente:CONCUR:2024},
series recognised by shuffle automata are closed under right derivatives.
We do not know how to establish this without shuffle-finite series,
since it is not clear how transitions for the new automaton should be defined.
\begin{proof}
    The proof is similar to the one of~\cref{lem:Hadamard closure properties}.
    Consider $f \in A := \poly \Q {f_1, \dots, f_k}_{\shuffle}$,
    where the latter algebra is differential.
    Consider the shuffle algebra $B$ generated by the $f_i$'s and their right derivatives,
    %
        $B := \poly \Q {f_1, \dots, f_k, \deriveright a f_1, \dots, \deriveright a f_k}_{\shuffle}$.
    %
    %
    \begin{claim*}
        For every series $g \in A$, we have $\deriveright a g \in B$.
    \end{claim*}
    \begin{proof}
        We proceed by induction on polynomial expressions.
        In the base case, we have $g = f_i \in A$
        and there is nothing to prove.
        In the case $g = g_1 + g_2$ with $g_1, g_2 \in A$,
        we have $\deriveright a g = \deriveright a g_1 + \deriveright a g_2$
        and the claim follows from the inductive assumption
        applied to $\deriveright a g_1, \deriveright a g_2$.
        Finally, consider $g = g_1 \shuffle g_2$ with $g_1, g_2 \in A$.
        By the product rule~\cref{eq:derive shuffle right},
        we have $\deriveright a g = \deriveright a g_1 \shuffle g_2 + g_1 \shuffle \deriveright a g_2$.
        By the inductive assumption, $\deriveright a g_1, \deriveright a g_2 \in B$,
        and thus the same holds for $\deriveright a g$,
    \end{proof}
    Since by assumption $f \in A$,
    by the claim $\deriveright a f \in B$.
    Left derivatives of the old generators $f_i$ are in $A \subseteq B$
    since $A$ is differential by assumption.
    It remains to show the same for the new generators $\deriveright a f_i$.
    So consider $\deriveleft b \deriveright a f_i$ for an arbitrary $b \in \Sigma$
    and we have to show that it is in $B$.
    By~\cref{lem:derive left right commutativity}, we have
    $\deriveleft a \deriveright b f_i = \deriveright b \deriveleft a f_i$.
    But by assumption $\deriveleft a f_i \in A$,
    therefore by the claim $\deriveright b \deriveleft a f_i \in B$, as required.
\end{proof}

We thus obtain the following synthesis
of the above-mentioned results from~\cite{Clemente:CONCUR:2024}
and~\cref{lem:shuffle finite right derivative}.
%
\begin{theorem}
    \label{thm:shuffle finite prevariety}
    The class of shuffle-finite series is an effective prevariety.
\end{theorem}
%

From~\cref{thm:commutativity for effective prevarieties},
\cref{thm:shuffle finite prevariety},
and the fact that the equality problem for shuffle-finite series is in~\TWOEXPSPACE~\cite[Theorem 1]{Clemente:CONCUR:2024},
we obtain the main result of the section.
\begin{theorem}
    \label{thm:shuffle finite commutativity problem}
    The commutativity problem for shuffle-finite series (equivalently, series recognised by shuffle automata)
    is decidable in~\TWOEXPSPACE.
\end{theorem}
\noindent
As directly inherited from the equality problem,
the complexity is doubly exponential in the number of generators $k$.

\subsection{Application: Multivariate constructible differentially algebraic power series}
\label{sec:CDA algebra}

In this section, we recall the definition of a rich class of multivariate series in commuting indeterminates
called \emph{constructible differentially algebraic} (\CDA).
Based on~\cref{thm:shuffle finite commutativity problem},
we show that they have an effective syntax,
addressing a problem left open in~\cite[Remark 27]{Clemente:CONCUR:2024}.
We begin in~\cref{sec:differential power series algebra} with some preliminaries,
then in~\cref{sec:CDA} we recall the definition of~\CDA,
and finally we show effectiveness in~\cref{sec:solvability of CDA equations}.

\subsubsection{Differential power series algebra}
\label{sec:differential power series algebra}

In the rest of the section, fix a dimension $d \in \N_{\geq 1}$
and consider $d$ pairwise commuting \emph{independent variables} $x = \tuple{x_1, \dots, x_d}$.
For a multi-index $n = \tuple{n_1, \dots, n_d} \in \N^d$,
we write $x^n$ for $x_1^{n_1} \cdots x_d^{n_d}$,
and $n!$ for $n_1! \cdots n_d!$.
An \emph{exponential multivariate power series} (in commuting variables) is a function $f : \N^d \to \Q$,
customarily written as $f = \sum_{n \in \N^d} f_n \cdot \frac {x^n} {n!}$
where $f_n \in \Q$ is the coefficient of term $x^n$.
We write the set of power series as $\powerseries \Q x$.
For example $1 + x_1 + x_1^2 + \cdots = \frac 1 {1 - x_1}$ is a power series.
Polynomials $\poly \Q x$ consist precisely of those power series
with \emph{finite support}, \ie, with $f_n = 0$ almost everywhere.
We endow the set of power series with the structure of a vector space
with zero $0$, scalar product $c \cdot f$, and addition $f + g$
all defined element-wise.
Therefore, the vector spaces of power series and number sequences are isomorphic.
However, multiplication is different.
While multiplication of number sequences is element-wise,
power series multiplication $f \cdot g$
extends to infinite supports
the familiar multiplication of polynomials
(sometimes called \emph{Cauchy product}).
Finally, the \emph{$j$-th partial derivative} $\partial {x_j} {}$ (for $1 \leq j \leq d$)
maps the power series $f$ to the power series
$\partial {x_j} f = \sum_{n \in \N^d} f_{n+e_j} \frac {x^n} {n!}$.
Notice that $\partial {x_j} {}$ is a derivation with respect to the product of power series,
since it satisfies the familiar Leibniz rule from calculus~(\cf~\cref{eq:Leibniz rule}):
\begin{align*}
    \partial {x_j} (f \cdot g) = (\partial {x_j} f) \cdot g + f \cdot (\partial {x_j} g).
\end{align*}
For example, $\partial {x_1} {\frac 1 {1 - x_1}} = \frac 1 {(1 - x_1)^2}$.
We have thus defined the \emph{differential power series algebra}
$\tuple{\powerseries \Q x; +, \cdot, (\partial {x_j})_{1 \leq j \leq d}}$,
where partial derivatives are pairwise commuting
$\partial {x_j} {} \circ \partial {x_h} {} = \partial {x_h} {} \circ \partial {x_j} {}$.
This will constitute the semantic background
for the developments in the rest of the section.

\subsubsection{\CDA~power series}
\label{sec:CDA}

A power series $f \in \powerseries \Q x$ is~\CDA\-
\cite{BergeronReutenauer:EJC:1990,BergeronSattler:TCS:1995}
if there exist $k \in \N_{\geq 1}$
and auxiliary power series $f_1, \dots, f_k \in \powerseries \Q x$ with $f = f_1$
satisfying a system of polynomial partial differential equations
\begin{align}
    \label{eq:CDA}
    \partial {x_j} f_i
        = p^{(j)}_i(f_1, \dots, f_k),
        \quad \forall 1 \leq i \leq k, 1 \leq j \leq d,
\end{align}
where $p^{(j)}_i \in \poly \Q y$ (with $y = \tuple{y_1, \dots, y_k}$)
for all $1 \leq i \leq k$ and $1 \leq j \leq d$.
We call systems of the form above~\emph{systems of autonomous \CDA~equations}.
\begin{remark}[Autonomous \vs~non-autonomous]
    \label{rem:nonautonomous CDA}
    We could allow the \rhs~of equations to contain the independent variables $x_j$'s,
    \ie, $p^{(j)}_i \in \poly \Q {x, y}$,
    obtaining \emph{non-autonomous} equations.
    This does not increase expressiveness
    since we can introduce fresh power series $g_j := x_j$
    together with equations $\partial {x_j} g_j = 1$,
    $\partial {x_h} g_j = 0$ for $h \neq j$,
    and initial conditions $g_j(0) = 0$.
    For this reason, we focus on autonomous equations.
\end{remark}
\noindent
\CDA~power series include all algebraic power series
(\ie, solutions of polynomial equations~\cite[Chapter III.16]{KuichSalomaa:1986}),
the exponential power series $e^x$,
the trigonometric power series $\sin x, \cos x$, etc...,
and are closed under natural operations,
such as scalar product, addition, multiplication, differentiation, and composition
(\cf~\cite[Lemma 21]{Clemente:CONCUR:2024}).
They find applications in combinatorics,
since the generating functions of a large class of finite structures are \CDA~\cite[Theorem 31]{Clemente:CONCUR:2024}.
Finally, the \CDA~equivalence problem is decidable in~\TWOEXPTIME,
thus of elementary complexity~\cite[Theorem 3]{Clemente:CONCUR:2024};
this should be contrasted with polyrec sequences,
for which the best upper bound is Ackermannian~\cref{sec:polyrec}.

\begin{example}[Binomial coefficients]
    \label{ex:binomial CDA}
    Consider the mixed generating power series of the binomial coefficients
    $f_1 := \sum_{n, k \in \N} {n \choose k} \cdot \frac {x_1^n} {n!} \cdot x_2^k = e^{x_1 \cdot (1 + x_2)}$.
    By introducing auxiliary power series $f_2 := x_1$ and $f_3 := x_2$,
    we obtain the \CDA~equations:
    \begin{align*}
        \arraycolsep=1pt
        \begin{array}{rl}
            \partial {x_1} f_1 &= f_1 \cdot (1 + f_3), \\
            \partial {x_2} f_1 &= f_1 \cdot f_2, \\
        \end{array}
        \qquad
        \begin{array}{rl}
            \partial {x_1} f_2 &= 1, \\
            \partial {x_2} f_2 &= 0, \\
        \end{array}
        \qquad
        \begin{array}{rl}
            \partial {x_1} f_3 &= 0, \\
            \partial {x_2} f_3 &= 1.
        \end{array}
    \end{align*}
\end{example}



\subsubsection{Solvability in power series of \CDA~equations}
\label{sec:solvability of CDA equations}

The definition of~\CDA~power series relies on the \emph{promise}
that~\cref{eq:CDA} does indeed have a power series solution.
This is a nontrivial problem, left open in~\cite[Remark 27]{Clemente:CONCUR:2024}.
\begin{problem}[\CDA~solvability problem]
    \label{problem:CDA solvability}
    The \emph{solvability problem} for \CDA~equations
    takes as input a set of \CDA~equations~\cref{eq:CDA}
    together with an initial condition $c \in \Q^k$,
    and it amounts to decide whether there is a power series solution $f \in \powerseries \Q x^k$ 
    extending the initial condition $f(0) = c$.
\end{problem}
\noindent
We remark that the initial condition is part of the input to the solvability problem.
In the \underline{univariate} case $d = 1$,
by the classic Picard-Lindelöf theorem~\cite[Theorem 3.1]{CoddingtonLevinson:1984},
\emph{every} initial condition extends to a solution,
and thus the solvability problem is trivial.
The multivariate case $d \geq 2$ is nontrivial,
as the following examples demonstrate.

\begin{example}[Unsolvable systems]
    \label{ex:CDA unsolvable system}
    A \CDA~system may not have any power series solution when $d \geq 2$,
    regardless of the initial condition.
    For instance, the following equations are unsatisfiable in power series:
    \begin{align*}
        \arraycolsep=1pt
        \begin{array}{rl}
            \partial {x_1} f &= 0, \\
            \partial {x_2} f &= g,
        \end{array}
        \quad \text{and} \quad
        \begin{array}{rl}
            \partial {x_1} g &= 1, \\
            \partial {x_2} g &= 1.
        \end{array}
    \end{align*}
    The equation $\partial {x_1} g = 1$ implies $g = x_1 + c(x_2)$ for some $c \in \powerseries \Q {x_2}$,
    and $\partial {x_2} g = 1$ implies $g = x_2 + d(x_1)$ for some $d \in \powerseries \Q {x_1}$.
    Taken together, we have $g = a + x_1 + x_2$ for some $a \in \Q$.
    The equation $\partial {x_1} f = 0$ implies $f \in \powerseries \Q {x_2}$,
    and thus $\partial {x_2} f = g \in \powerseries \Q {x_2}$, which is a contradiction.
\end{example}

\begin{example}[Sensitivity to the inital condition]
    \label{rem:not uniformably solvable}
    Existence of power series solutions
    may depend on the initial value when $d \geq 2$.
    %
    Consider equations
    \begin{align*}
        \arraycolsep=1pt
        \begin{array}{rl}
            \partial {x_1} f &= f+g, \\
            \partial {x_2} f &= 0,
        \end{array}
        \quad \text{and} \quad
        \begin{array}{rl}
            \partial {x_1} g &= 0, \\
            \partial {x_2} g &= g.
        \end{array}
    \end{align*}
    Thus $f \in \powerseries \Q {x_1}$,
    and $g = a \cdot e^{x_2 + b} \in \powerseries \Q {x_2}$, for some $a, b \in \Q$.
    Since $g = \partial {x_1} f - f \in \powerseries \Q {x_2}$,
    there are no solutions with $g(0) = a \cdot e^b \neq 0$.
    When $g(0) = 0$ (iff $a = 0$), we have the family of solutions
    $g = 0$ and $f = c \cdot e^{{x_1}+d}$ for every $c, d \in \Q$.
\end{example}

Solvability in power series of polynomial partial differential equations is undecidable, as we now recall.
%

\begin{remark}
    \label{rem:undecidability of solvability}
    %
    A seminal result of Denef and Lipshitz shows that
    solvability in power series for polynomial partial differential equations is \emph{undecidable}~\cite[Theorem 4.11]{DenefLipshitz:1984}.
    In fact, this holds already for a \emph{single}, \emph{linear} partial differential equation in \emph{one unknown},
    with polynomial coefficients in $\poly \Q x$.
    In order to gain insight into the shape of the undecidable equations,
    we present their reduction.
    Let $P(y_1, \dots y_d) \in \poly \Z {y_1, \dots, y_d}$
    be a multivariate polynomial with integer coefficients.
    Consider the non-autonomous differential equation
    \begin{align}
        \label{eq:undecidable linear PDE}
        P(x_1 \cdot \partial {x_1} {}, \dots, x_d \cdot \partial {x_d} {}) f 
        = \frac 1 {1 - x_1} \cdots \frac 1 {1 - x_d},
    \end{align}
    with initial condition $f(0) = P(0)^{-1}$ (assuming $P(0) \neq 0$).
    For example, if $P(y_1, y_2) = 1 - 2 \cdot y_1 + y_2^2$
    then we would get the equation
    $(1 - 2 \cdot x_1 \cdot \partial {x_1} {} + (x_2 \cdot \partial {x_2} {})^2) f = \frac 1 {1 - x_1} \frac 1 {1 - x_2}$
    where $(x_2 \cdot \partial {x_2} {})^2$ is the derivation $x_2 \cdot \partial {x_2} {}$ applied twice.
    Equation~\cref{eq:undecidable linear PDE} has a solution in power series
    iff $P = 0$ has no nonnegative integer solutions:
    Indeed, the \rhs~is just $\sum_{n \in \N^d} x^n$
    and if $f \in \powerseries \C x$ is the power series $f = \sum_{n \in \N^d} f_n \cdot x^n$,
    then the \lhs~equals $\sum_{n \in \N^d} f_n \cdot P(n) \cdot x^n$.
    As a consequence of Hilbert 10th problem (\cf~\cite{Matijasevic:2003}),
    whether $P = 0$ has a nonnegative integer solution is undecidable 
    (already for $d = 9$),
    and thus solvability in power series for \cref{eq:undecidable linear PDE} is undecidable as well.

    Note that~\cref{eq:undecidable linear PDE} is not in the \CDA~format.
    The first reason is that it is non-autonomous since it contains occurrences of the independent variables $x_j$'s.
    This is inessential and it can be addressed by~\cref{rem:nonautonomous CDA}.
    %
    %
    The rational power series $\frac 1 {1 - x_j}$ on the \rhs~of~\cref{eq:undecidable linear PDE}
    can be eliminated by multiplying both sides of the equation by $(1 - x_1) \cdots (1 - x_d)$.
    %
    The second reason is more fundamental:
    Both $\partial {x_j} f$ and $\partial {x_h} f$ (with $j \neq h$) appear in the same equation,
    which in \CDA~is not allowed.
\end{remark}

As a consequence of~\cref{rem:undecidability of solvability},
in order to decide~\cref{problem:CDA solvability} we need to exploit the~\CDA~format.
To this end we develop a connection between
power series solutions of \CDA~equations~\cref{eq:CDA} and shuffle automata,
in a way analogous to~\cref{sec:consistency of polyrec equations}.


Consider an alphabet of noncommuting indeterminates $\Sigma = \set {a_1, \dots, a_d}$.
The differential algebra of power series
$\tuple{\powerseries \Q x; {+}, {\cdot}, (\partial {x_j})_{1 \leq j \leq d}}$
and the differential shuffle algebra of \emph{commutative} series
$\tuple{\commseries \Q \Sigma; {+}, {\shuffle}, (\deriveleft {a_j})_{1 \leq j \leq d}}$.
The isomorphism maps a commutative series $f \in \commseries \Q \Sigma$
to the power series $\widetilde f = \sum_{n \in \N^d} \widetilde f_n \cdot \frac {x^n} {n!} \in \powerseries \Q x$
with coefficients $\widetilde f_{n_1, \dots, n_d} := f(a_1^{n_1} \cdots a_d^{n_d})$ for every $n_1, \dots, n_d \in \N$.
In other words, this is a bijective mapping preserving the algebraic and differential structure, \ie,
$\widetilde \zero = 0$, $\widetilde {1 \cdot \e} = 1$, and
%
\begin{align}
    \label{eq:comm iso}
    \begin{array}{rl}
    \widetilde {c \cdot f}
        &= c \cdot \widetilde f, \\
    \widetilde {f + g}
        &= \widetilde f + \widetilde g,
    \end{array}
    \qquad
    \begin{array}{rl}
    \widetilde {f \shuffle g}
        &= \widetilde f \cdot \widetilde g, \\
    \widetilde {\deriveleft {a_j} f}
        &= \partial {x_j} {\widetilde f},
    \end{array}
\end{align}
for all $f, g \in \commseries \Q \Sigma$, $c \in \Q$, and $1 \leq j \leq d$.
For a system of \CDA~equations~\cref{eq:CDA},
together with an initial condition $c = (c_1, \dots, c_k) \in \Q^k$,
we construct its \emph{companion shuffle automaton} $\tuple{\Sigma, X, F, \Delta}$,
which has nonterminals $X = \tuple{X_1, \dots, X_k}$,
output function $F X_i := c_i$ for all $1 \leq i \leq k$,
and transitions
\begin{align}
    \label{eq:companion shuffle automaton}
    \Delta_{a_j} X_i := p^{(j)}_i(X_1, \dots, X_k),
    \quad \forall 1 \leq i \leq k, 1 \leq j \leq d.
\end{align}
Let $f_i := \sem {X_i}$ be the series recognised by nonterminal $X_i$ of the automaton, for all $1 \leq i \leq k$,
and write $f = \tuple{f_1, \dots, f_k}$.
\begin{lemma}
    \label{lem:CDA solvability}
    The initial condition $c$ extends to a power series solution of the \CDA~equations if, and only if,
    $f$ is a tuple of commutative series.    
\end{lemma}
\begin{proof}
    On the one hand,
    if $f$ is a tuple of commutative series with initial condition $\coefficient \e f = c$ solving~\cref{eq:companion shuffle automaton},
    then, by following the isomorphism,
    $\widetilde f \in \powerseries \Q x^k$
    is a tuple of power series solutions of~\cref{eq:CDA} with initial condition $f(0) = c$.
    Indeed, by the properties of the semantics of shuffle automata~(\cref{lem:properties of semantics - shuffle automata}),
    component $f_i$ satisfies differential equations
    \begin{align*}
        \deriveleft {a_j} f_i = p^{(j)}_i(f_1, \dots, f_k),
        \quad \forall 1 \leq i \leq k, 1 \leq j \leq d.
    \end{align*}
    By applying the isomorphism to both sides of the equation above,
    and by~\cref{eq:comm iso}, we get
    \begin{align*}
        \partial {x_j} {\widetilde {f_i}}
            = \widetilde {\deriveleft {a_j} {f_i}}
            = \widetilde {\left(p^{(j)}_i(f_1, \dots, f_k)\right)}
            = p^{(j)}_i(\widetilde {f_1}, \dots, \widetilde {f_k}),
    \end{align*}
    which is just~\cref{eq:CDA}.
    On the other hand,
    if $g \in \powerseries \Q x^k$ is a tuple of power series solutions
    of~\cref{eq:CDA} with initial condition $g(0) = c$,
    then by following the inverse of the isomorphism 
    we obtain a tuple of commutative series $h \in (\commseries \Q \Sigma)^k$
    \st~$\widetilde h = g$ and $\coefficient \e h = c$,
    solving equations~\cref{eq:companion shuffle automaton}.
    Since equations~\cref{eq:companion shuffle automaton}
    admit exactly one series solution with $c$ as initial condition,
    $h = f$ is the semantics of the shuffle automaton.
    Consequently, $f$ is a tuple of commutative series.
\end{proof}

\begin{example}[Uniformably solvable systems]
    Consider the \CDA~equations from~\cref{ex:binomial CDA}, and fix the initial condition $c := \tuple{1, 0, 0}$.
    The companion shuffle automaton is the one from~\cref{ex:binomial shuffle automaton}.
    It can be verified that $\Delta_{a_1} \circ \Delta_{a_2} = \Delta_{a_2} \circ \Delta_{a_1}$.
    It suffices to consider the generators $X_1, X_2, X_3$:
    \begin{align*}
        \Delta_{a_1} \Delta_{a_2} X_1
            &= \Delta_{a_1} (X_1 \cdot X_2)
            = X_1 \cdot (1 + X_3) \cdot X_2 + X_1, \\
        \Delta_{a_2} \Delta_{a_1} X_1
            &= \Delta_{a_2} (X_1 \cdot (1 + X_3))
            = X_1 \cdot X_2 \cdot (1 + X_3) + X_1, \\
        \Delta_{a_1} \Delta_{a_2} X_2
            &= \Delta_{a_2} \Delta_{a_1} X_2 = 0, \\
        \Delta_{a_1} \Delta_{a_2} X_3
            &= \Delta_{a_2} \Delta_{a_1} X_3 = 0.
    \end{align*}
    Thus, \emph{every} initial condition extends to a solution,
    \ie, \cref{ex:binomial CDA} is \emph{uniformably solvable}.
    %
    %
\end{example}

We now have all ingredients to solve~\cref{problem:CDA solvability}.
For \CDA~equations~\cref{eq:CDA} and initial condition $c \in \Q^k$
we build their companion shuffle automaton~(\cf~\cref{eq:companion shuffle automaton}),
and check that all series $\sem {X_1}, \dots, \sem {X_k}$ recognised by its nonterminals are commutative.
The latter is decidable in~\TWOEXPSPACE~by~\cref{thm:shuffle finite commutativity problem}.
Correctness follows from~\cref{lem:CDA solvability}.
\begin{theorem}
    \label{thm:decidability of CDA solvability}
    The \CDA~solvability problem~(\cref{problem:CDA solvability}) is decidable in \TWOEXPSPACE.
\end{theorem}
\section{Extensions}
\label{sec:conclusions}

\subsection{Infiltration product, automata, and series}
\label{sec:infiltration}

%
Besides Hadamard and shuffle products,
in the literature one finds another product,
called \emph{infiltration}.
It is denoted by $f \infiltration g$ and is defined coinductively as follows\-
\cite[Def.~8.1]{BasoldHansenPinRutten:MSCS:2017}:
\begin{align}
    \tag{${\infiltration}$-$\e$}
    \label{eq:infiltration:base}
    \coefficient \e {(f \infiltration g)}
        &:= f_\e \cdot g_\e, \\
    \tag{${\infiltration}$-$\deriveleft a$}
    \label{eq:infiltration:step}
        \deriveleft a (f \infiltration g)
        &:= \deriveleft a f \infiltration g
            + f \infiltration \deriveleft a g
            + \deriveleft a f \infiltration \deriveleft a g,
        \ \forall a \in \Sigma.
\end{align}
This is an associative and commutative operation,
with the same identity $1 \cdot \e$ as the shuffle product.
It has been introduced by Chen, Fox, and Lyndon~\cite{ChenFoxLyndon:AM:1958}
in the context of the \emph{free differential calculus} of Fox~\cite{Fox:AM:1953}.
More recently, it has been considered in a coalgebraic setting~\cite[Sec.~8]{BasoldHansenPinRutten:MSCS:2017}.
From a concurrency perspective,
it models the \emph{synchronising interleaving}
where processes are allowed (but not forced) to jointly perform the input actions
(thanks to the term $\deriveleft a f \infiltration \deriveleft a g$ in~\cref{eq:infiltration:step}).
\Eg, in the interleaving semantics we have $ab \shuffle a = 2aab + aba$,
but in the synchronising interleaving $ab \infiltration a = 2aab + aba + ab$.

We observe that the infiltration product falls under the scope of our techniques.
In a manner analogous to the Hadamard and shuffle products,
one can define \emph{infiltration automata},
and corresponding \emph{infiltration-finite series},
which are novel models not previously investigated.
One can then prove that they constitute an effective prevariety,
and thus the commutativity problem is decidable for this class.
More details in~\cref{app:infiltration}.

\subsection{Mixed product, automata, and series}

One can go one step further, and consider a more general model
whether the Hadamard, shuffle, and infiltration semantics are combined.
Consider a finite input alphabet $\Sigma$ partitioned into three subalphabets
$\Sigma = \Sigma_{\hadamard} \cupdot \Sigma_{\shuffle} \cupdot \Sigma_{\infiltration}$.
Based on this partitioning, consider the binary operation on series ``$\parallel$''
(whose dependency on the partitioning is hidden for readability)
\st~$\coefficient \e (f \parallel g) = \coefficient f \e \cdot \coefficient g \e$,
and for every $a \in \Sigma$,
\begin{align*}
    \deriveleft a (f \parallel g) =
    \left\{\begin{array}{ll}
        \deriveleft a f \parallel \deriveleft a g
            &\text{if } a \in \Sigma_{\hadamard}, \\
        \deriveleft a f \parallel g + f \parallel \deriveleft a g
            &\text{if } a \in \Sigma_{\shuffle}, \\
        \deriveleft a f \parallel g + f \parallel \deriveleft a g + \deriveleft a f \parallel \deriveleft a g
            &\text{if } a \in \Sigma_{\infiltration}.
    \end{array}\right.
\end{align*}
This models actions in $\Sigma_{\hadamard}$ with a synchronising semantics,
actions in $\Sigma_{\shuffle}$ with an interleaving semantics,
and in $\Sigma_{\infiltration}$ with a synchronising interleaving semantics.
This yields an associative and commutative operation ``$\parallel$'',
specialising to the Hadamard $\Sigma = \Sigma_{\hadamard}$,
shuffle $\Sigma = \Sigma_{\shuffle}$,
or infiltration product $\Sigma = \Sigma_{\infiltration}$.
Techniques developed in this paper can be applied to this more general product,
obtaining a notion of \emph{$\parallel$-finite series}
(which therefore generalise the shuffle-finite, Hadamard-finite, and infiltration-finite series),
and corresponding \emph{$\parallel$-automata},
recognising the same class.
Finally, $\parallel$-finite series are an effective prevariety,
and thus equality and commutativity are decidable for them.



\subsection{Effective algebraic varieties of commutative series}
\label{sec:effective algebraic varieties}

We can generalise the commutativity problem
and show that the output functions leading to a commutative semantics
form an \emph{effective algebraic variety} (in the sense of algebraic geometry), as we now explain.
Fix a Hadamard, shuffle, or infiltration automaton $\tuple{\Sigma, X, \Delta}$,
where we have omitted the output function $F$.
We consider a slightly more general class of outputs functions $F \in \C^X$,
mapping nonterminals to complex numbers. 
%
%
Fix an initial configuration $\alpha \in \poly \Q X$
and let $\FF \subseteq \C^X$ be the set of output functions giving rise to a commutative semantics, i.e.,
\begin{align*}
    \FF := \setof{F \in \C^X}{\sem {\alpha} \text{with output function $F$ is commutative}}.
\end{align*}
$\FF$ is the set of common zeroes of the polynomials
\begin{align}
    \label{eq:commutativity polynomials}
    P := \setof {\Delta_u \alpha - \Delta_v \alpha \in \poly \Q X}{u, v \in \Sigma^*, u \sim v},
\end{align}
where recall that $u \sim v$ means that $u, v$ are permutation equivalent.
In other words, $\FF$ is an \emph{algebraic variety}~\cite[§2, Ch.~1]{CoxLittleOShea:Ideals:2015}
and commutativity amounts to decide whether a given output function $F$ is in $\FF$
(membership problem).
%
Thanks to decidability of commutativity,
there is a computable bound $N \in \N$ \st\-
$\FF$ is the set of common zeroes of the \emph{finite} and \emph{computable} set of polynomials $P_N \subseteq P$
obtained by considering words of length $\leq N$.
%
For Hadamard automata, $N$ is Ackermannian,
and for shuffle automata it is doubly exponential.
As a consequence, $\FF$ is an \emph{effective} algebraic variety,
since we can decide $F \in \FF$ by checking whether all polynomials in $P_N$ vanish on $F$.
We can use $P_N$ to solve the following two variants of the commutativity problem:
%
\begin{enumerate}
    \item 
    Decide whether \emph{there exists} an output function $F\in\C^X$
    giving rise to a commutative semantics.

    \item 
    Decide whether \emph{every} output function $F\in\C^X$
    gives rise to a commutative semantics.
\end{enumerate}
The second problem is equivalent to $P_N = \set 0$.
Thanks to \emph{Hilbert's Nullstellensatz}~\cite[Theorem 1, §1, Ch.~4]{CoxLittleOShea:Ideals:2015},
the first problem is equivalent to whether the \emph{ideal} generated by $P_N$ does not contain $1$.
The latter can be established with standard techniques
from computational algebraic geometry,
such as \emph{Gröbner bases}~\cite[Ch.~2]{CoxLittleOShea:Ideals:2015}.
More details in~\cref{app:varieties}.

\section*{Acknowledgements}
We warmly thank Arka Ghosh, Aliaume Lopez, and Filip Mazowiecki 
for valuable discussions during the preparation of this work.

\bibliographystyle{plainurl}
\bibliography{literature}

\appendices

\section{Preliminaries}
\label{app:preliminaries}

In this section, we provide more details about~\cref{sec:preliminaries}.

\leftRightComm*
\begin{proof}
    This follows from associativity of concatenation of finite words:
    Indeed, for any series $f \in \series \Q \Sigma$ and word $w \in \Sigma^*$ we have,
    \begin{align*}
        \coefficient w {\deriveleft u \deriveright v f}
            &= \coefficient {uw} {\deriveright v f}
            = \coefficient {uwv} f, \quad \text{ and } \\
        \coefficient w {\deriveright v \deriveleft u f}
            &= \coefficient {wv} {\deriveright v f}
            = \coefficient {uwv} f. \qedhere
    \end{align*}
\end{proof}

\section{Commutativity}
\label{app:commutativity}

In this section, we provide more details about~\cref{sec:commutativity}.

\lemShuffleReflectsCommutativity*
\noindent
In the proof of the lemma we will need a simple fact about the shuffle product.
Recall that a \emph{zero divisor} in a commutative ring is an element $x \in R$
\st~there exists a nonzero $y \in R$ with $x \cdot y = 0$.
An \emph{integral domain} is a commutative ring without nontrivial zero divisors.
The ring of series with the shuffle product is an integral domain\-
\cite[Theorem 3.2]{Ree:AM:1958},
and thus it does not have nontrivial zero divisors.

\begin{proof}
    Clearly if $f = \zero$ then $ab \shuffle f = \zero$ is commutative.
    For the other direction, assume that $f$ is not the zero series.
    Since the shuffle product does not have nontrivial zero divisors and $ab$ is not zero,
    also $g := ab \shuffle f$ is not zero and thus its support is nonempty.
    Let $w$ be any word in the support of $g$, thus $g_w \neq 0$.
    Such a word is of the form $w = x a y b z$ for some words $x, y, z \in \Sigma^*$
    \st~$f_{xyz} \neq 0$.
    Now consider the word $w' := x b y a z$ obtained from $w$ by swapping $a$ and $b$.
    By the definition of shuffle product, every word in the support of $g$ features an $a$ followed by a $b$,
    but $x, y, z$ do not contain neither $a$ or $b$.
    Thus $w'$ is not in the support of $g$, i.e., $g_{w'} = 0$.
    Since $w \sim w'$ are commutatively equivalent, $g$ is not commutative.
\end{proof}

\section{Hadamard automata and series}
\label{app:hadamard}

In this section, we provide more details about~\cref{sec:hadamard}.

\subsection{Preliminaries}

\lemRightDerivationHadamardEndo*
\begin{proof}
    Linearity follows directly from the definition of $\deriveright a$.
    Equation~\cref{eq:derive hadamard right} could be proved coinductively.
    In fact, a direct argument suffices since the Hadamard product acts element-wise:
    \begin{align*}
        \coefficient w {\deriveright a (f \hadamard g)}
        &= \coefficient {wa} {(f \hadamard g)}
        = \coefficient {wa} f \cdot \coefficient {wa} g = \\
        &= \coefficient w \deriveright a f \cdot \coefficient w \deriveright a g
        = \coefficient w (\deriveright a f \hadamard \deriveright a g). \qedhere
    \end{align*}
\end{proof}

\subsection{Hadamard automata}

\lemPropertiesofSemantics*
\begin{proof}
    The property $\sem 0 = \zero$ is trivial,
    and $\sem 1 = \one$ holds by definition:
    $\coefficient w \sem 1 = F \Delta_w 1 = F 1 = 1$.
    The property~\cref{eq:hadamard:sem:derivation}
    holds by definition of the semantics.
    The other proofs follow natural coinductive arguments.
    Properties~\cref{eq:hadamard:sem:scalar-product} and~\cref{eq:hadamard:sem:sum}
    follow from linearity of $\deriveleft a$ and $\Delta_a$.
    For property~\cref{eq:hadamard:sem:product}, we argue coinductively.
    Clearly the constant terms agree:
    \begin{align*}
        \coefficient \e \sem {\alpha \cdot \beta}
        &= F (\alpha \cdot \beta)
        = F \alpha \cdot F \beta
        = \coefficient \e \sem \alpha \cdot \coefficient \e \sem \beta = \\
        &= \coefficient \e (\sem \alpha \hadamard \sem \beta).
    \end{align*}
    For the coinductive step, we have 
    (for readability, we write $\alpha^a$ for $\Delta_a \alpha$ when $\alpha \in \poly \Q X$)
    %
    \begin{align*}
        \deriveleft a \sem {\alpha \cdot \beta}
            &= \sem {\Delta_a (\alpha \cdot \beta)}
            = \sem {\alpha^a \cdot \beta^a}
            = \deriveleft a \sem \alpha \hadamard \deriveleft a \sem \beta = \\
            &= \deriveleft a (\sem \alpha \hadamard \sem \beta). \qedhere
    \end{align*}
\end{proof}

\subsection{Hadamard-finite series}

The following lemma constitutes our working definition of Hadamard-finite series.
We provide a full proof.
\lemHadamardFiniteWorkingDefinition*
\begin{proof}
    If $f$ is Hadamard finite, then the two conditions hold by definition.
    \Wlg~we can take $f = f_1$ to be one of the generators
    since we can just add $f$ to the generators,
    without changing the validity of the two conditions.

    For the other direction, assume that there are generators $f_1, \dots, f_k$ with $f = f_1$ satisfying the two conditions.
    We need to show that the finitely generated algebra $A := \poly \Q {f_1, \dots, f_k}_{\hadamard}$ is closed under left derivatives.
    To this end, consider a series $g \in A$.
    There is a polynomial $p \in \poly \Q k$ \st~$g = p(f_1, \dots, f_k)$.
    Take now the left derivative of both sides.
    Since they are homomorphisms of $A$, we have
    $\deriveleft a g = \deriveleft a (p(f_1, \dots, f_k)) = p(\deriveleft a f_1, \dots, \deriveleft a f_k)$.
    By the second condition, $\deriveleft a f_1, \dots, \deriveleft a f_k \in A$,
    and thus $\deriveleft a g \in A$, as required.
\end{proof}

\lemHadamardAutomataAndSeries*
\begin{proof}
    %
    For the ``only if'' direction, let $f = \sem {X_1}$
    be recognised by a Hadamard automaton $\tuple{\Sigma, X, F, \Delta}$
    with nonterminals $X = \tuple{X_1, \dots, X_k}$.
    We show that $f$ is Hadamard finite
    by applying~\cref{lem:Hadamard finite - working definition}.
    Consider series $f_i := \sem {X_i}$ for all $1 \leq i \leq k$
    generating the Hadamard algebra $A := \poly \Q {f_1,\dots, f_k}_{\hadamard}$.
    Clearly $f \in A$.
    Now consider generator $f_i$ and input symbol $a \in \Sigma$,
    and we need to show $\deriveleft a f_i \in A$.
    By~\cref{lem:Hadamard automata - properties of the semantics},
    the semantics is a homomorphism of difference algebras, and thus
    \begin{align*}
        \deriveleft a f_i
        &= \deriveleft a \sem {X_i}
        = \sem {\Delta_a X_i}
        = (\Delta_a X_i)(\sem {X_1}, \dots \sem {X_k}) = \\
        &= (\Delta_a X_i)(f_1, \dots, f_k) \in A,
    \end{align*}
    as required.

    For the ``if'' direction, let $f \in \series \Q \Sigma$ be Hadamard finite.
    By~\cref{lem:Hadamard finite - working definition},
    there are series $f_1, \dots, f_k \in \series \Q \Sigma$ with $f = f_1$
    generating a Hadamard algebra $A = \poly \Q {f_1, \dots, f_k}_{\hadamard}$
    \st, for every generator $f_i$ and input symbol $a \in \Sigma$,
    there is a polynomial $p_i^a \in \poly \Q {X_1, \dots, X_k}$
    \st~$\deriveleft a f_i = p_i^a(f_1, \dots, f_k)$.
    We have all ingredients to build a Hadamard automaton recognising $f_1$.
    Consider the automaton $\tuple{\Sigma, X, F, \Delta}$
    with nonterminals $X := \tuple{X_1, \dots, X_k}$,
    output mapping $F(X_i) := \coefficient \e f_i$
    and transitions
    \begin{align*}
        \Delta_a(X_i) := p_i^a(X_1, \dots, X_k),
        \quad \text{for all } 1 \leq i \leq k \text{ and } a \in \Sigma.
    \end{align*}
    The proof is concluded by showing $\sem {X_i} = f_i$ for every $1 \leq i \leq k$.
    We argue coinductively.
    First, the two series have the same constant term by construction,
    \begin{align*}
        \coefficient \e \sem{X_i} = F \Delta_\e X_i = F X_i = \coefficient \e f_i.
    \end{align*}
    Second, the tails agree
    \begin{align*}
        \deriveleft a \sem {X_i} &=
            &\text{(by~\cref{lem:Hadamard automata - properties of the semantics})} \\
        &= \sem {\Delta_a X_i} = 
            &\text{(def.~of $\Delta_a$)} \\
        &= \sem {p_i^a(X_1, \dots, X_k)} = 
            &\text{(by~\cref{lem:Hadamard automata - properties of the semantics})} \\
        &= p_i^a(\sem {X_1}, \dots, \sem {X_k}) = 
            &\text{(by~coind.)} \\
        &= p_i^a(f_1, \dots, f_k) = 
            &\text{(def.~of $p_i^a$)} \\
        &= \deriveleft a f_i. \qedhere
    \end{align*}
\end{proof}

\lemHadamardClosureProperties*
\begin{proof}
    Closure under right derivative was shown already in the main text.
    Closure under scalar product is trivial.
    Regarding closure under addition and Hadamard product,
    consider two Hadamard-finite series $f, g \in \series \Q \Sigma$
    \st~$f \in A := \poly \Q {f_1, \dots, f_k}_{\hadamard}$,
    $g \in B := \poly \Q {g_1, \dots, g_\ell}_{\hadamard}$,
    and $A, B$ are closed under left derivatives.
    Consider the finitely generated Hadamard algebra
    obtained by concatenating the generators,
    \begin{align*}
        C := \poly \Q {f_1, \dots, f_k, g_1, \dots, g_\ell}_{\hadamard}.
    \end{align*}
    Clearly $f + g, f \hadamard g \in C$.
    It remains to show that $C$ is closed under left derivatives.
    But this is clear, since both $A$ and $B$ are.
    
    Finally, consider closure under left derivative.
    We have $\deriveleft a f \in A$ since $f \in A$, $A$ is closed under left derivatives,
    and $\deriveleft a$ is a homomorphism of difference algebras.
\end{proof}

We now prove that Hadamard-finite series over disjoint alphabets are closed under shuffle product.

\lemHadamardShuffle*
\begin{proof}
    Recall that elements in a finitely generated Hadamard algebra over alphabet $\Sigma$
    are either the constant $\zero$, the constant $\one := \sum_{w \in \Sigma^*} 1 \cdot w$ (the algebra identity),
    a generator, or are built inductively from scalar product, sum, and Hadamard product.

    Let now $\Sigma, \Gamma$ be disjoint alphabets.
    For a series $f \in \series \Q \Sigma$,
    let ``$f \shuffle \_$'' be the operation $\series \Q \Gamma \to \series \Q {\Sigma \cup \Gamma}$
    of shuffling with $f$ applied to series in $\series \Q \Gamma$.
    In the following claims,
    we show that ``$f \shuffle \_$'' distributes over the primitives
    of finitely generated Hadamard algebras $\subseteq \series \Q \Gamma$.
    
    Distributivity over scalar product is true in general,
    and does not require the assumption that $\Sigma, \Gamma$ are disjoint.
    \begin{claim*}
        For every series $f \in \series \Q \Sigma$ and $g \in \series \Q \Gamma$, we have
        \begin{align}
            f \shuffle (c \cdot g) = c \cdot (f \shuffle g).
        \end{align}
    \end{claim*}

    Distributivity over addition is true in general,
    and does not require the assumption that $\Sigma, \Gamma$ are disjoint.
    \begin{claim*}
        For every series $f \in \series \Q \Sigma$
        and $g, h \in \series \Q \Gamma$, we have
        \begin{align}
            f \shuffle (g + h) = f \shuffle g + f \shuffle h.
        \end{align}
    \end{claim*}
    Finally, we have distributivity over the Hadamard product,
    where we rely in an essential way on the assumption
    that the alphabets $\Sigma, \Gamma$ are disjoint.
    In the rest of the proof, for brevity we write $f^a$ for $\deriveleft a f$.
    \begin{claim*}
        For every series $f \in \series \Q \Sigma$ and $g, h \in \series \Q \Gamma$
        \st~$f(\e) \in \set{0, 1}$ we have
        %
        \begin{align}
            \label{eq:hadamard-shuffle}
            f \shuffle (g \hadamard h) = (f \shuffle g) \hadamard (f \shuffle h).
        \end{align}
    \end{claim*}
    \begin{proof}[Proof of the claim]
        We argue coinductively.
        The constant terms of the two sides of~\cref{eq:hadamard-shuffle} are
        $f(\e)g(\e)h(\e)$, resp., $f(\e)^2g(\e) h(\e)$,
        and these two quantities are equal when $f(\e) \in \set{0, 1}$.
        Now we take the left derivative of both sides \wrt~$a \in \Sigma \cup \Gamma$,
        and apply the coinductive hypothesis:
        \begin{align*}
            &\deriveleft a (f \shuffle (g \hadamard h)) = \\
                &= f^a \shuffle (g \hadamard h) + f \shuffle (g^a \hadamard h^a) = \\
                &= (f^a \shuffle g) \hadamard (f^a \shuffle h)
                 + (f \shuffle g^a) \hadamard (f \shuffle h^a), \\[1ex]
            &\deriveleft a ((f \shuffle g) \hadamard (f \shuffle h)) = \\
                &= (f^a \shuffle g + f \shuffle g^a) \hadamard (f^a \shuffle h + f \shuffle h^a).
        \end{align*}
        In general the two terms above are not equal.
        However, since the alphabets are disjoint, they are equal in both of the two cases of interest below:
        \begin{enumerate}
            \item $a \in \Sigma$: in this case $g^a = h^a = 0$;
            \item $a \in \Gamma$: in this case $f^a = 0$. \qedhere
        \end{enumerate}
    \end{proof}
    The previous three claims come together in the following last claim.
    \begin{claim*}
        For every series $f \in \series \Q \Sigma$ \st~$f(\e) \in \set{0, 1}$,
        polynomial $p \in \poly \Q {y_1, \dots, y_k}$ without constant term $p(0) = 0$,
        and series $g_1, \dots, g_k \in \series \Q \Gamma$, we have
        \begin{align}
            \label{eq:last claim}
            f \shuffle p(g_1, \dots, g_k) = p(f \shuffle g_1, \dots, f \shuffle g_k).
        \end{align}
    \end{claim*}
    \begin{proof}
        Write $p(g_1, \dots, g_k)$ as
        \begin{align*}
            p(g_1, \dots, g_k)
            = \sum_{n_1, \dots, n_k \in \N} c_{n_1, \dots, n_k} \cdot g_1^{n_1} \hadamard \cdots \hadamard g_k^{n_k},
        \end{align*}
        where $c_{n_1, \dots, n_k} \in \Q$ with $c_{0, \dots, 0} = 0$.
        By $g_i^{n_i}$ we mean the repeated Hadamard product $g_i \hadamard \cdots \hadamard g_i$ of $g_i$ with itself $n_i$ times.
        By the previous claims, we have
        \begin{align*}
            &f \shuffle p(g_1, \dots, g_k) = \\
            &= f \shuffle \sum_{n_1, \dots, n_k \in \N} c_{n_1, \dots, n_k} \cdot g_1^{n_1} \hadamard \cdots \hadamard g_k^{n_k} = \\
            &= \sum_{n_1, \dots, n_k \in \N} f \shuffle (c_{n_1, \dots, n_k} \cdot g_1^{n_1} \hadamard \cdots \hadamard g_k^{n_k}) = \\
            &= \sum_{n_1, \dots, n_k \in \N} c_{n_1, \dots, n_k} \cdot f \shuffle (g_1^{n_1} \hadamard \cdots \hadamard g_k^{n_k}) = \\
            &= \sum_{n_1, \dots, n_k \in \N} c_{n_1, \dots, n_k} \cdot (f \shuffle g_1)^{n_1} \hadamard \cdots \hadamard (f \shuffle g_k)^{n_k} = \\
            &= p(f \shuffle g_1, \dots, f \shuffle g_k). \qedhere
        \end{align*}
    \end{proof}
    Now suppose $f = f_1$ and $g = g_1$ are Hadamard finite,
    thus belonging to finitely generated difference Hadamard algebras
    \begin{align*}
        A := \poly \Q {f_1, \dots, f_m}_{\hadamard}, \quad \text{resp.}, \quad
        B := \poly \Q {g_1, \dots, g_n}_{\hadamard}.
    \end{align*}
    By rescaling if necessary, we can assume without loss of generality
    that all constant terms of the generators are either zero or one $f_i(\e), g_j(\e) \in \set{0, 1}$.
    Now consider the finitely generated Hadamard algebra
    \begin{align*}
        C := \polyof \Q {f_i \shuffle g_j} {1 \leq i \leq m, 1 \leq j \leq n}_{\hadamard}.
    \end{align*}
    Clearly $f_1 \shuffle g_1$ belongs to $C$.
    We need to show that $C$ is closed under left derivatives.
    Recall that $f_i^a = p_i^a(f_1, \dots, f_m)$ and $g_j^a = q_j^a(g_1, \dots, g_n)$
    for polynomials $p_i^a, q_j^a$ without constant term $p_i^a(0) = q_j^a(0) = 0$.
    Take now the left derivative of an arbitrary generator of $C$, and we have
    \begin{align*}
        &\deriveleft a (f_i \shuffle g_j)
            = f_i^a \shuffle g_j + f_i \shuffle g_j^a = \\
            &= p_i^a(f_1, \dots, f_m) \shuffle g_j + f_i \shuffle q_j^a(g_1, \dots, g_n) = \\
            &= p_i^a(f_1 \shuffle g_j, \dots, f_m \shuffle g_j) + q_j^a(f_i \shuffle g_1, \dots, f_i \shuffle g_n),
    \end{align*}
    where in the last equality we have applied the last claim~\cref{eq:last claim}.
    It follows that $\deriveleft a (f_i \shuffle g_j) \in C$, as required.
\end{proof}

\subsection{Polynomial vs.~Hadamard automata}

In this section, we give more details about the relationship
between polynomial and Hadamard automata,
expanding on~\cref{rem:Hadamard and polynomial automata}.
A \emph{polynomial automaton}~\cite{BenediktDuffSharadWorrell:PolyAut:2017} is a tuple
\begin{align}
    \label{eq:polynomial automaton}
    A = \tuple{k, \Sigma, Q, q_I, \Delta, F}
\end{align}
where $k \in \N$ is the \emph{dimension},
$\Sigma$ is a finite \emph{input alphabet},
$Q = \Q^k$ is the set of \emph{states},
$q_I \in Q$ is the \emph{initial state},
$\Delta : \Sigma \to Q \to Q$ is the polynomial \emph{transition function},
and $F : Q \to \Q$ is the polynomial \emph{output function}.
For every input symbol $a \in \Sigma$,
the polynomial map $\Delta^a : Q^k \to Q^k$
is presented as a tuple of polynomials $\Delta^a = \tuple{\Delta^a_1, \dots, \Delta^a_k} \in \poly \Q k^k$,
inducing a polynomial action on states
\begin{align*}
    q \cdot a := \Delta^a(q) = \tuple{\Delta^a_1(q), \dots, \Delta^a_k(q)} \in Q,
    \quad\text{for all } q \in Q.
\end{align*}
Similarly, the polynomial output function $F$ is presented as a polynomial $F \in \poly \Q k$
inducing the corresponding polynomial map
$F(q) = F(q_1, \dots, q_k) \in \Q$ for every $q = \tuple{q_1, \dots, q_k} \in Q$.
The action of $\Sigma$ is extended to words $w \in \Sigma^*$ homomorphically:
$q \cdot \e := q$ and $q \cdot (a \cdot w) := (q \cdot a) \cdot w$.
The \emph{semantics} of a state $q \in Q$ is the series defined as follows:
\begin{align*}
    \sem q &\in \series \Q \Sigma \\
    \sem q (w) &:= F(q \cdot w), \quad \text{for every } w \in \Sigma^*.
\end{align*}
The semantics of the automaton $A$ is the series recognised by the initial state
$\sem A {} := \sem {q_I} {}$.
\begin{lemma}
    A series $f \in \series \Q \Sigma$ is recognisable by a polynomial automaton
    if, and only if, its reversal $f^R$ is Hadamard finite.
\end{lemma}
\noindent
Recall that $f^R$ is the series that maps a word $w \in \Sigma^*$ to $f(w^R)$,
where $w^R = a_n \cdots a_1$ is the reversal of $w = a_1 \cdots a_n$.
\begin{proof}
    For the ``only if'' direction,
    let $A$ be the polynomial automaton recognising $f$ as in~\cref{eq:polynomial automaton}.
    For every $1 \leq i \leq k$,
    consider the series $g_i \in \series \Q \Sigma$ defined as
    \begin{align*}
        g_i(w) := \pi_i(q_I \cdot w^R),
        \quad \text{for all } w \in \Sigma^*.
    \end{align*}
    In other words, $g_i(w)$ is the value of component $i$
    after reading $w^R$ from the initial state.
    Consider now the Hadamard algebra
    \begin{align*}
        A := \poly \Q {g_1, \dots, g_k}_{\hadamard} \subseteq \series \Q \Sigma.
    \end{align*}
    Thanks to~\cref{lem:Hadamard finite - working definition},
    this part of the proof is concluded by the following two claims.
    \begin{claim*}
        $f^R$ is in $A$.
    \end{claim*}
    \begin{proof}
        By definition of $f$, we have $f = \sem {q_I} {}$.
        Consequently for every $w \in \Sigma^*$ we have
        \begin{align*}
            f^R(w)
                &= F(q_I \cdot w^R)
                = F(\pi_1(q_I \cdot w^R), \dots, \pi_k(q_I \cdot w^R)) = \\
                &= F(g_1(w), \dots, g_k(w)),
        \end{align*}
        and thus $f^R = F(g_1, \dots, g_k)$ is in $A$,
        as required.
    \end{proof}

    \begin{claim*}
        For every $a \in \Sigma$ and $1 \leq i \leq k$,
        we have $\deriveleft a g_i \in A$.
    \end{claim*}
    \begin{proof}
        By the definition of $g_i$, for every $w \in \Sigma^*$ we have
        \begin{align*}
            (\deriveleft a g_i)(w)
                &= g_i(a \cdot w)
                = \pi_i(q_I \cdot w^R \cdot a)
                = \Delta^a_i(q_I \cdot w^R) = \\
                &= \Delta^a_i(g_1(w), \dots, g_k(w)),
        \end{align*}
        and thus $\deriveleft a g_i = \Delta^a_i(g_1, \dots, g_k)$
        is in $A$, as required.
    \end{proof}

    For the ``if'' direction, let $g \in \series \Q \Sigma$
    belong to a finitely generated Hadamard algebra $A := \poly \Q {g_1, \dots, g_k}_{\hadamard}$
    closed under left derivatives $\deriveleft a$ (with $a \in \Sigma$).
    Since $g \in A$, we can write $g = F(g_1, \dots, g_k)$
    for some polynomial $F \in \poly \Q k$.
    For every $1 \leq i \leq k$ and $a \in \Sigma$,
    the series $\deriveleft a g_i$ is in $A$ and thus it can be written as
    $\Delta^a_i(g_1, \dots, g_k)$ for some polynomial $\Delta^a_i \in \poly \Q k$.
    Finally, let $q_I := \tuple{g_1(\e), \dots, g_k(\e)} \in \Q^k$ be the tuple of constant terms of the generators.
    This provides us with the data required to construct a polynomial automaton
    of dimension $k$, as in~\cref{eq:polynomial automaton}.
    The proof is concluded with the following two claims.
    \begin{claim*}
        For every $1 \leq i \leq k$ and $w \in \Sigma^*$, we have
        \begin{align*}
            g_i(w) = \pi_i(q_I \cdot w^R).
        \end{align*}
    \end{claim*}
    \begin{proof}
        We proceed by induction on $w$.
        In the base case $w = \e$, by definition we have
        $\pi_i(q_I \cdot \e^R) = \pi_i(q_I) = g_i(\e)$.
        In the inductive case, we have
        \begin{align*}
            g_i(a \cdot w)
                &=\coefficient w \deriveleft a g_i
                = \coefficient w \Delta^a_i(g_1, \dots, g_k) = \\
                &= \Delta^a_i(g_1(w), \dots, g_k(w)) = \\
                &= \Delta^a_i(\pi_1(q_I \cdot w^R), \dots, \pi_k(q_I \cdot w^R)) = \\
                &= \Delta^a_i(q_I \cdot w^R)
                = \pi_i(q_I \cdot w^R \cdot a) = \\
                &= \pi_i (q_I \cdot (a \cdot w)^R).
            \qedhere
        \end{align*}
    \end{proof}
    \begin{claim*}
        We have $g = \sem {q_I}^R$.
    \end{claim*}
    \begin{proof}
        Consider an arbitrary $w \in \Sigma^*$,
        and we need to show $g(w) = F(q_I \cdot w^R)$.
        Thanks to the previous claim, we have
        \begin{align*}
            g(w)
                &= F(g_1(w), \dots, g_k(w)) = \\
                &= F(\pi_1(q_I \cdot w^R), \dots, \pi_k(q_I \cdot w^R)) = \\
                &= F(q_I \cdot w^R).
            \qedhere
        \end{align*}
    \end{proof}
    \qedhere
\end{proof}

\subsection{Number sequences}
\label{app:polyrec}

In the main text we have mentioned closure properties of polyrec sequences,
namely closure under scalar product, sum, Hadamard product, shifts, sections, and diagonals.
The last two operations have not been defined yet.
A \emph{section} of a {$d$-variate} sequence $f : \N^d \to \Q$
is an {$e$-variate} sequence with $e < d$ which is obtained from $f$
by fixing one or more coordinates to a constant value~\cite[Definition 3.1]{Lipshitz:D-finite:JA:1989}.
For instance, if $f(n_1, n_2)$ is a bivariate sequence, then $g(n_1) := f(n_1, 7)$ is a section thereof,
obtained by fixing the second coordinate to be $n_2 := 7$.
Since sections are defined element-wise,
they commute with the operations of scalar product, sum, Hadamard product,
and shifts not involving the coordinates being fixed.
It follows that sections of polyrec sequences are polyrec.

\begin{example}
    For instance, if $\shift 1 f = \one - f^2$,
    then its section $g$ satisfies the same equation $\shift 1 g = \one - g^2$,
    preserving the polyrec format
    (information about $\shift 2 f$ is discarded when fixing the second coordinate).
\end{example}

The other operation that we have not yet defined is that of taking diagonals.
A \emph{diagonal} of a {$d$-variate} sequence $f : \N^d \to \Q$
is an $e$-variate sequence with $e < d$ which is obtained from $f$
by requiring a subset of the coordinates to be equal
and then by projecting them to a single coordinate\-
\cite[Definition 2.6]{Lipshitz:D-finite:JA:1989}.
For instance, the sequence $h : \N^2 \to \Q$ \st
\begin{align*}
    h(n_1, n_3) := f(n_1, n_1, n_3),
    \quad\text{for all } n_1, n_3 \in \N,
\end{align*}
is the diagonal of $f : \N^3 \to \Q$
obtained by identifying the first two coordinates.
Again, this is an element-wise operation,
and thus it commutes with the operations of scalar product, sum, Hadamard product,
and shifts not involving the coordinates being identified.
Consequently, diagonals of polyrec sequences are polyrec.

\begin{example}
    For instance, if $\shift 1 f = \one - f^2$, $\shift 2 f = f^3$, and $\shift 3 f = \one + 2 \cdot f$
    then the diagonal $h$ satisfies $\shift 1 h = (\one - h^2)^3$ and $\shift 2 h = \one + 2 \cdot h$,
    preserving the polyrec format.
    %
    %
    (By composing in the other order, we get $\shift 1 h = 1 - (h^3)^2$, but one equation suffices:
    If $f$ exists, then it will satisfy both equations.)
\end{example}

\section{Shuffle automata and series}

In this section, we provide more details about~\cref{sec:shuffle}.

\lemRightDerivationShuffleDer*
\begin{proof}
    We have already observed linearity in~\cref{sec:preliminaries}.
    Regarding~\cref{eq:derive shuffle right}, we prove it by coinduction.
    First of all, the constant term of both sides is the same, since
    \begin{align*}
        \coefficient \e {\deriveright a (f \shuffle g)}
        &= \coefficient a {(f \shuffle g)} = f_a \cdot g_\e + f_\e \cdot g_a, \text{ and } \\ 
        \coefficient \e {(\deriveright a f \shuffle g + f \shuffle \deriveright a g)}
        &= \coefficient \e {(\deriveright a f \shuffle g)} + \coefficient \e {(f \shuffle \deriveright a g)} = \\
        &= \coefficient \e {\deriveright a f} \cdot \coefficient \e g + \coefficient \e f \cdot \coefficient \e \deriveright a g = \\
        &= f_a \cdot g_\e + f_\e \cdot g_a.
    \end{align*}
    The proof is concluded by showing that, for every $b \in \Sigma$,
    the $b$-left derivative $\deriveleft b$ of both sides is the same:
    \begin{align*}
        \deriveleft b \deriveright a (f \shuffle g)
        = \deriveleft b (\deriveright a f \shuffle g + f \shuffle \deriveright a g).
    \end{align*}
    We indeed have
    \begin{align*}
        &\deriveleft b \deriveright a (f \shuffle g) = \\
        &= \deriveright a \deriveleft b (f \shuffle g)
            && \text{(by~\cref{lem:derive left right commutativity})} \\ 
        &= \deriveright a (\deriveleft b f \shuffle g + f \shuffle \deriveleft b g)
            && \text{(by~\cref{eq:shuffle:step})} \\ 
        &= \deriveright a (\deriveleft b f \shuffle g) + \deriveright a (f \shuffle \deriveleft b g)
            && \text{(by lin.)} \\
        &= \deriveright a \deriveleft b f \shuffle g + \deriveleft b f \shuffle \deriveright a g + {} \\
        & \qquad + \deriveright a f \shuffle \deriveleft b g + f \shuffle \deriveright a \deriveleft b g
            && \text{(by coind.)} \\
        &= \deriveleft b \deriveright a f \shuffle g + \deriveleft b f \shuffle \deriveright a g + {} \\
        & \qquad + \deriveright a f \shuffle \deriveleft b g + f \shuffle \deriveleft b \deriveright a g
            && \text{(by~\cref{lem:derive left right commutativity})} \\
        &= \deriveleft b (\deriveright a f \shuffle g) + \deriveleft b (f \shuffle \deriveright a g)
            && \text{(by~\cref{eq:shuffle:step})} \\ 
        &= \deriveleft b (\deriveright a f \shuffle g + f \shuffle \deriveright a g)
            && \text{(by lin.)}. \qedhere
    \end{align*}
\end{proof}

Since we prove that shuffle series are a prevariety,
one may wonder whether they are even a variety, in the sense of Reutenauer.
Recall that varieties are prevarieties
satisfying the following additional closure condition:
\begin{enumerate}[label=\textbf{(V.\arabic*)}]
    \setcounter{enumi}{2}
    \item \label{prevariety:C}
    For every series $f \in \SS_\Sigma$
    and algebra homomorphism%
    \footnote{
        The corresponding requirement in \cite[Sec.~III.1]{Reutenauer_1980}
        demands closure \wrt~algebra homomorphisms of the form $\varphi \in \series \Q \Gamma \to \series \Q \Sigma$.
        However $f \circ \varphi$ may not be defined when $\varphi$ produces series with infinite supports.
        For instance take $f = \varphi(a) = 1 + a + a^2 + \cdots$.
        Then $(f \circ \varphi)(a) = \inner f {\varphi(a)} = 1 + 1 + \cdots$ is not defined.
        For this reason $\varphi$ needs to be restricted to be a homomorphism of series with finite supports.
    }
    $\varphi : \ncpoly \Q \Gamma \to \ncpoly \Q \Sigma$,
    the series $f \circ \varphi$ is in $\SS_\Gamma$.
\end{enumerate}
In the next remark we rule out this possibility,
thanks to a simple growth argument.

\begin{remark}[Shuffle-finite series are not a variety]
    The class of shuffle-finite series is not a variety of series.
    This can be shown by a simple growth argument.
    Consider a unary input alphabet $\Sigma = \set a$.
    Univariate shuffle-finite series $f = \sum_{n \in \N} f_n \cdot a^n$
    are in bijective correspondence with univariate \CDA~power series
    $\widetilde f(x) = \sum_{n \in \N} f_n \cdot \frac {x^n} {n!}$~\cite[Lemma 25]{Clemente:CONCUR:2024}.
    For the latter class it is known that $f_n \in \bigO {\alpha^n \cdot n!}$ for some constant $\alpha > 0$
    \cite[Theorem 3.(i)]{BergeronReutenauer:EJC:1990}.
    Consider the shuffle-finite series \st~$f(\e) := 1$
    and $\deriveleft a f := f \shuffle f$.
    It can be checked that $f(a^n) = n!$,
    and thus $\widetilde f(x) = \sum_{n \in \N} f_n \cdot \frac {x^n} {n!}$
    with $f_n = n!$.
    Take the homomorphism $\varphi : \Sigma^* \to \Sigma^*$
    defined by $\varphi(a) := aa$.
    We have that the composition series $g := f \circ \varphi$
    satisfies $g(a^n) = f(a^{2n}) = (2n)!$.
    The corresponding power series is $\widetilde g(x) = \sum_{n \in \N} (2n)! \cdot \frac {x^n} {n!}$,
    which grows too fast in order to be \CDA.
    Consequently $g$ is not shuffle finite and thus shuffle-finite series are not a variety.
\end{remark}
\section{Infiltration automata and series}
\label{app:infiltration}

In this section, we provide more details about~\cref{sec:infiltration}.
We introduce a model of weighted computation, called \emph{infiltration automata},
and the associated class of \emph{infiltration-finite series}.
All results in this section are new.

\subsection{Infiltration product}

Recall that the \emph{infiltration product} of two series $f, g \in \series \Q \Sigma$,
denoted by $f \infiltration g$,
is defined coinductively as follows:
\begin{align}
    \tag{${\infiltration}$-$\e$}
    \label{eq:infiltration:base'}
    \coefficient \e {(f \infiltration g)}
        &= f_\e \cdot g_\e, \\
    \tag{${\infiltration}$-$\deriveleft a$}
    \label{eq:infiltration:step'}
        \deriveleft a (f \infiltration g)
        &= \deriveleft a f \infiltration g
            + f \infiltration \deriveleft a g
            + \deriveleft a f \infiltration \deriveleft a g,
        \ \forall a \in \Sigma.
\end{align}
To see why this uniquely defines a series $f \infiltration g$,
we can reason by induction on the length of words.
The first rule gives us the value for the constant term $\coefficient \e (f \infiltration g)$.
For nonempty words, we have
\begin{align*}
    &\coefficient {a \cdot w} (f \infiltration g) = \\
    &= \coefficient w (\deriveleft a (f \infiltration g)) = \\
    &= \coefficient w (\deriveleft a f \infiltration g
        + f \infiltration \deriveleft a g
        + \deriveleft a f \infiltration \deriveleft a g) = \\
    &= \coefficient w (\deriveleft a f \infiltration g)
    + \coefficient w (f \infiltration \deriveleft a g)
    + \coefficient w (\deriveleft a f \infiltration \deriveleft a g),
\end{align*}
where the three coefficients for word $w$ are known by the inductive assumption.
For example,
\begin{align*}
    \coefficient {ab} (f \infiltration g)
    &= \coefficient b (f^a \infiltration g + f \infiltration g^a + f^a \infiltration g^a) = \\
    &= \coefficient \e (f^{ab} \infiltration g + f^a \infiltration g^b + f^{ab} \infiltration g^b + \\
    &\quad f^b \infiltration g^a + f \infiltration g^{ab} + f^b \infiltration g^{ab} + \\
    &\quad f^{ab} \infiltration g^a + f^a \infiltration g^{ab} + f^{ab} \infiltration g^{ab}) \\
    &= f_{ab} \cdot g + f_a \cdot g_b + f_{ab} \cdot g_b + \\
    &\quad f_b \cdot g_a + f \cdot g_{ab} + f_b \cdot g_{ab} + \\
    &\quad f_{ab} \cdot g_a + f_a \cdot g_{ab} + f_{ab} \cdot g_{ab}.
\end{align*}
For brevity, we use the convention that $f^a$ denotes $\deriveleft a f$,
and similarly for longer words.
Infiltration product is an associative and commutative operation on series,
with identity the series $1 \cdot \e$ (the same identity as for the shuffle product).
We call the resulting structure $\series \Q \Sigma_{\infiltration} := \tuple{\series \Q \Sigma; {+}, {\infiltration}}$
the \emph{infiltration algebra} of series.
%


\begin{remark}
    We have chosen a coinductive approach to the definition of infiltration product
    since it brings to the foreground the similarity with Hadamard and shuffle products.
    Alternatively, one can define the infiltration product on finite words by induction on their length,
    and then lift it to series by linearity and continuity~\cite[page 126]{Lothaire:CUP:1997}.
\end{remark}

\begin{remark}[Differential algebra of series]
    We note that the infiltration algebra has a natural differential structure.
    Indeed, for every $a \in \Sigma$, consider the mapping $\sigma_a$ on series \st
    \begin{align}
        \sigma_a f := f + \deriveleft a f,
        \text{ for every } f \in \series \Q \Sigma.
    \end{align}
    In short, we have $\sigma_a := 1 + \deriveleft a$.
    It is easy, albeit lengthy, to check that $\sigma_a$ is an infiltration algebra endomorphism:
    \begin{align*}
        \sigma_a (c \cdot f)
            &= c \cdot f + \deriveleft a (c \cdot f) = \\
            &= c \cdot f + c \cdot \deriveleft a f = \\
            &= c \cdot (f + \deriveleft a f) = \\
            &= c \cdot \sigma_a f, \\
        \sigma_a (f + g)
            &= f + g + \deriveleft a (f + g) = \\
            &= f + \deriveleft a f + g + \deriveleft a g = \\
            &= \sigma_a f + \sigma_a g. \\
        \sigma_a (f \infiltration g)
            &= f \infiltration g + \deriveleft a (f \infiltration g) = \\
            &= f \infiltration g + f^a \infiltration g + \sigma_a f \infiltration g^a = \\
            &= \sigma_a f \infiltration g + \sigma_a f \infiltration g^a = \\
            &= \sigma_a f \infiltration (g + g^a) = \\
            &= \sigma_a f \infiltration \sigma_a g.
    \end{align*}
    Using the definition of $\sigma_a$,
    the rule~\cref{eq:infiltration:step} can be rewritten as follows:
    \begin{align}
        \tag{${\infiltration}$-$\deriveleft a$'}
        \label{eq:infiltration:step''}
            \deriveleft a (f \infiltration g)
            &:= \deriveleft a f \infiltration g
                + \sigma_a f \infiltration \deriveleft a g
            \ \forall a \in \Sigma.
    \end{align}
    Using differential algebra terminology,
    this makes $\deriveleft a$ a \emph{$\sigma_a$-derivation} of the infiltration algebra,
    and thus $\tuple{\series \Q \Sigma; {+}, {\infiltration}, {\sigma}, \deriveleft {}}$
    is a $\sigma$-differential algebra.
    This should be compared to the fact that $\deriveleft a$ is a derivation
    (\ie, for the identity endomorphism)
    of the shuffle algebra.
    This helps explaining why the infiltration product behaves more similarly to
    the shuffle product than to the Hadamard one.
\end{remark}

\subsection{Infiltration automata}
\label{sec:infiltration automata}

\paragraph{Syntax}

An \emph{infiltration automaton} is a tuple
$A = \tuple{\Sigma, X, F, \Delta}$
where $\Sigma$ is a finite \emph{input alphabet},
$X = \set{X_1,  \dots X_k}$ is a finite set of \emph{nonterminal symbols},
$F : X \to \Q$ is the \emph{output function},
and $\Delta : \Sigma \to X \to \poly \Q X$ is the \emph{transition function}.
Thus, at the syntactic level infiltration automata contain the same information as Hadamard and shuffle automata.
However, their semantics makes them different.

\paragraph{Semantics}
A \emph{configuration} is a polynomial $\alpha \in \poly \Q N$.
For every input symbol $a \in \Sigma$,
the transition function $\Delta_a : X \to \poly \Q X$
is extended to the unique linear function $\Delta_a : \poly \Q X \to \poly \Q X$
satisfying the following product rule, for all $\alpha, \beta \in \poly \Q X$:
%
\begin{align}
    \label{eq:infiltration}
    \tag{infiltration}
    f(\alpha \cdot \beta) =
        f(\alpha) \cdot \beta +
        \alpha \cdot f(\beta) +
        f(\alpha) \cdot f(\beta).
\end{align}
%
%
In turn, this allows us to extend transitions from single letters
to all finite input words homomorphically:
For every configuration $\alpha \in \poly \Q X$,
input word $w \in \Sigma^*$ and letter $a \in \Sigma$,
we have $\Delta_\e \alpha := \alpha$
and $\Delta_{a \cdot w} \alpha := \Delta_w \Delta_a \alpha$.
We have all ingredients to define the \emph{semantics} of a configuration $\alpha \in \poly \Q X$,
which is the series $\sem \alpha \in \series \Q \Sigma$ defined as follows:
\begin{align*}
    \sem \alpha_w &:= F \Delta_w \alpha,
    \text{ for all } w \in \Sigma^*.
\end{align*}
Here, $F$ is extended from nonterminals to configurations homomorphically
$F(\alpha \cdot \beta) = F \alpha \cdot F \beta$ (i.e., polynomial evaluation).
The semantics of an infiltration automaton $A$
is the series recognised by its distinguished nonterminal $\sem {X_1}$.

\begin{remark}[Differential algebra of polynomials]
    For every $a \in \Sigma$,
    consider the mapping $S_a : \poly \Q X \to \poly \Q X$ \st
    \begin{align}
        S_a(\alpha) := \alpha + \Delta_a(\alpha),
        \text{ for all } \alpha \in \poly \Q X.
    \end{align}
    In short, $S_a := 1 + \Delta_a$.
    It can be checked that this is an endomorphism of the polynomial ring and that
    \begin{align}
        \Delta_a(\alpha \cdot \beta) = \Delta_a \alpha \cdot \beta + S_a(\alpha) \cdot \Delta_a \beta.
    \end{align}
    Thus $\Delta_a$ is an $S_a$-derivation of the polynomial ring for every $a \in \Sigma$,
    and thus $\tuple{\poly \Q X; {+}, {\cdot}, (S_a)_{a \in \Sigma}, (\Delta_a)_{a \in \Sigma}}$
    is an $S$-differential algebra.
\end{remark}

\begin{lemma}[Properties of the semantics]
    \label{lem:infiltration automata - properties of the semantics}
    The semantic function $\sem \_$ is a differential homomorphism
    from the $S$-differential algebra of polynomials
    to the $\sigma$-differential infiltration algebra of series.
    In other words, 
    \begin{align}
        \label{eq:infiltration:sem:scalar-product}
        \sem {c \cdot \alpha}
            &= c \cdot \sem \alpha, \\
        \label{eq:infiltration:sem:sum}
            \sem {\alpha + \beta}
            &= \sem \alpha + \sem \beta, \\
        \label{eq:infiltration:sem:product}
            \sem {\alpha \cdot \beta}
            &= \sem \alpha \infiltration \sem \beta, \\
        \label{eq:infiltration:sem:endomorphism}
            \sem {S_a \alpha}
                &= \sigma_a \sem \alpha, \\
            \label{eq:infiltration:sem:derivation}
        \sem {\Delta_a \alpha}
            &= \deriveleft a \sem \alpha.
    \end{align}
\end{lemma}
\begin{proof}
    Property~\cref{eq:infiltration:sem:derivation}
    holds by definition of the semantics.
    The other proofs follow natural coinductive arguments.
    Properties~\cref{eq:infiltration:sem:scalar-product} and~\cref{eq:infiltration:sem:sum}
    follow from linearity of $\deriveleft a$ and $\Delta_a$.
    From these, \cref{eq:infiltration:sem:endomorphism} follows easily:
    \begin{align*}
        \sem {S_a \alpha}
        &= \sem {\alpha + \Delta_a \alpha}
        = \sem \alpha + \sem {\Delta_a \alpha}
        = \sem \alpha + \deriveleft a \sem \alpha
        = \sigma_a \sem \alpha.
    \end{align*}

    For the last property~\cref{eq:infiltration:sem:product}, we proceed coinductively.
    Clearly, the constant term of $\sem {\alpha \cdot \beta}$ and $\sem \alpha \infiltration \sem \beta$
    is the same, namely $F \alpha \cdot F \beta$.
    For the coinductive step, we show that the left derivatives
    of $\sem {\alpha \cdot \beta}$ and $\sem \alpha \infiltration \sem \beta$ coincide.
    For readability, we write $\alpha^a$ for $\Delta_a \alpha$ when $\alpha \in \poly \Q X$
    and $f^a$ for $\deriveleft a f$ when $f \in \series \Q \Sigma$.
    \begin{align*}
        \deriveleft a \sem {\alpha \cdot \beta}
            &= \sem {\Delta_a (\alpha \cdot \beta)} = \\
            &= \sem {\alpha^a \cdot \beta + S_a \alpha \cdot \beta^a} = \\
            &= \sem {\alpha^a \cdot \beta} + \sem {S_a \alpha \cdot \beta^a} = \\
            &= \sem {\alpha^a} \infiltration \sem \beta + \sem {S_a \alpha} \infiltration \sem {\beta^a} = \\
            &= (\deriveleft a {\sem \alpha}) \infiltration \sem \beta + (\sigma_a \sem \alpha )\infiltration \deriveleft a \sem \beta = \\
            &= \deriveleft a (\sem \alpha \infiltration \sem \beta). \qedhere
    \end{align*}
\end{proof}

\subsection{Infiltration-finite series}

In order to gain insights about the class of series recognised by infiltration automata,
it is convenient to introduce the same class from a different, semantic point of view.
For series $f_1, \dots, f_k \in \series \Q \Sigma$,
we denote by $\poly \Q {f_1, \dots, f_k}_{\infiltration}$
the infiltration subalgebra generated by the $f_i$'s.
%
%
An infiltration algebra is \emph{differential}
if it is closed under $\deriveleft a$, for all $a \in \Sigma$.

\begin{definition}
    A series is \emph{infiltration finite} if it belongs to a
    finitely generated differential infiltration algebra of series.
\end{definition}
The following lemma unpacks the definition above
and provides our working definition for infiltration finite series.
Its proof is similar to~\cref{lem:Hadamard finite - working definition}.
\begin{lemma}
    \label{lem:infiltration finite - working definition}
    A series $f \in \series \Q \Sigma$ is infiltration finite
    iff there are generators $f_1, \dots, f_k \in \series \Q \Sigma$ \st:
    \begin{enumerate}
        \item $f \in \poly \Q {f_1, \dots, f_k}_{\infiltration}$ and
        \item $\deriveleft a f_i \in \poly \Q {f_1, \dots, f_k}_{\infiltration}$
        for every $a \in \Sigma$ and $1 \leq i \leq k$.
    \end{enumerate}
    Moreover, we can assume \wlg~that $f = f_1$.
\end{lemma}


The following characterisation connects infiltration automata and infiltration-finite series.
Thanks to it, we can regard infiltration automata as finite syntactic presentations of infiltration-finite series.
\begin{restatable}{lemma}{lemInfiltrationAutomataInfiltrationFinite}
    \label{lem:infiltration automata = infiltration finite}
    A series is recognised by an infiltration automaton
    if, and only if, it is infiltration finite.
\end{restatable}
\begin{proof}
   The proof is analogous to that of~\cref{lem:Hadamard automata and series},
   relying on~\cref{lem:infiltration automata - properties of the semantics,lem:infiltration finite - working definition}.
\end{proof}

\subsection{Closure properties}

In this section, we show that the series recognised by infiltration automata
(equivalently, the infiltration-finite series by~\cref{lem:infiltration automata = infiltration finite})
have several pleasant closure properties.
There are two ways of presenting and proving such closure properties.
One way proceeds by syntactically manipulating infiltration automata.
We have opted for another, more elegant approach
relying on infiltration finiteness.

\begin{restatable}{lemma}{lemInfiltrationClosure}
    \label{lem:infiltration closure}
    The class of infiltration-finite series
    is an effective differential infiltration algebra.
    In particular, infiltration-finite series are effectively closed under
    scalar product $c \cdot f$,
    sum $f + g$,
    infiltration product $f \infiltration g$,
    and left derivative $\deriveleft a f$.
\end{restatable}

\begin{proof}
    We show closure under infiltration product,
    the other properties being similar.
    Let $f, g$ be infiltration finite
    with generators $f_1, \dots, f_k$, resp., $g_1, \dots, g_\ell$
    (by applying~\cref{lem:infiltration finite - working definition}).
    Consider the finitely generated infiltration algebra
    \begin{align*}
        A := \poly \Q {f_1, \dots, f_k, g_1, \dots, g_\ell}_{\infiltration}.
    \end{align*}
    Clearly $f \infiltration g \in A$.
    Moreover $A$ is closed under $\deriveleft a$, for all $a \in \Sigma$,
    since $\deriveleft a f_i$ is in $\poly \Q {f_1, \dots, f_k}_{\infiltration}$ for all $1 \leq i \leq k$
    and $\deriveleft a g_j$ is in $\poly \Q {g_1, \dots, g_\ell}_{\infiltration}$ for all $1 \leq j \leq \ell$.
\end{proof}


\begin{lemma}
    \label{lem:derive right infiltration}
    The right derivative operator $\deriveright a$ (with $a \in \Sigma$)
    is linear and satisfies the following property:
    \begin{align}
        \label{eq:derive right infiltration}
        \deriveright a (f \infiltration g)
        &= \deriveright a f \infiltration g
            + f \infiltration \deriveright a g
            + \deriveright a f \infiltration \deriveright a g,
        \quad \forall a \in \Sigma.
    \end{align}
\end{lemma}
\noindent
In other words, $\deriveright a$ is a $\sigma^R_a$-derivation
for the endomorphism $\sigma^R_a := 1 + \deriveright a$,
since we can rewrite~\cref{eq:derive right infiltration} as
\begin{align}
    \label{eq:derive right infiltration'}
    \deriveright a (f \infiltration g)
    &= \deriveright a f \infiltration g
        + \sigma^R_a f \infiltration \deriveright a g,
    \quad \forall a \in \Sigma.
\end{align}

\begin{proof}
    We proceed by coinduction.
    The constant terms of the two sides coincide:
    \begin{align*}
        &\coefficient \e (\deriveright a (f \infiltration g)) = \\
        &= \coefficient a (f \infiltration g) = \\
        &= \coefficient \e (\deriveleft a (f \infiltration g)) = \\
        &= \coefficient \e (f^a \infiltration g + f \infiltration g^a + f^a \infiltration g^a) = \\
        &= \coefficient \e (f^a \infiltration g) + \coefficient \e (f \infiltration g^a) + \coefficient \e (f^a \infiltration g^a) = \\
        &= f_a \cdot g_\e + f_\e \cdot g_a + f_a \cdot g_a = \\
        &= \coefficient \e (\deriveright a f \infiltration g) + \coefficient \e (f \infiltration \deriveright a g) + \coefficient \e (\deriveright a f \infiltration \deriveright a g) = \\
        &= \coefficient \e (\deriveright a f \infiltration g + f \infiltration \deriveright a g + \deriveright a f \infiltration \deriveright a g).
    \end{align*}
    Now consider an input symbol $b \in \Sigma$, and we have
    \begin{align*}
        &\deriveleft b \deriveright a (f \infiltration g) =
            &\text{(\cref{lem:derive left right commutativity})} \\
        &= \deriveright a \deriveleft b (f \infiltration g) =
            &\text{(def.~$\infiltration$)} \\
        &= \deriveright a (f^b \infiltration g + f \infiltration g^b + f^b \infiltration g^b) =
            &\text{(lin.~$\deriveright a$)} \\
        &= \deriveright a (f^b \infiltration g)
            + \deriveright a (f \infiltration g^b)
            + \deriveright a (f^b \infiltration g^b) =
            &\text{(coind.)} \\
        &=  \deriveright a f^b \infiltration g + f^b \infiltration \deriveright a g + \deriveright a f^b \infiltration \deriveright a g + \\
        &\ +\deriveright a f \infiltration g^b + f \infiltration \deriveright a g^b + \deriveright a f \infiltration \deriveright a g^b + \\
        &\ +\deriveright a f^b \infiltration g^b + f^b \infiltration \deriveright a g^b + \deriveright a f^b \infiltration \deriveright a g^b =
            &\text{(\cref{lem:derive left right commutativity})} \\
        &=  (\deriveright a f)^b \infiltration g + f^b \infiltration \deriveright a g + (\deriveright a f)^b \infiltration \deriveright a g + \\
        &\ +\deriveright a f \infiltration g^b + f \infiltration (\deriveright a g)^b + \deriveright a f \infiltration (\deriveright a g)^b + \\
        &\ +(\deriveright a f)^b \infiltration g^b + f^b \infiltration (\deriveright a g)^b + (\deriveright a f)^b \infiltration (\deriveright a g)^b =
            &\text{(def.~$\infiltration$)} \\
        &= \deriveleft b (\deriveright a f \infiltration g)
            + \deriveleft b (f \infiltration \deriveright a g)
            + \deriveleft b (\deriveright a f \infiltration \deriveright a g) =
            &\text{(lin.~$\deriveleft b$)} \\
        &= \deriveleft b (\deriveright a f \infiltration g
        + f \infiltration \deriveright a g
        + \deriveright a f \infiltration \deriveright a g). \qedhere
    \end{align*}
\end{proof}

\begin{restatable}{lemma}{lemInfiltrationClosureRightDerivative}
    \label{lem:infiltration closure - right derivative}
    The class of infiltration-finite series
    is effectively closed under right derivative $\deriveright a f$.
\end{restatable}

\begin{proof}
    Let $f$ be an infiltration-finite series and fix an input symbol $a \in \Sigma$.
    By~\cref{lem:infiltration finite - working definition},
    there are generators $f_1, \dots, f_k$ with $f = f_1$ generating the $\deriveleft {}$-closed infiltration algebra
    \begin{align*}
        A := \poly \Q {f_1, \dots, f_k}_{\infiltration}.
    \end{align*}
    %
    We adjoin new generators $\deriveright a f_1, \dots, \deriveright a f_k$,
    obtaining the larger infiltration algebra
    \begin{align*}
        B := \poly \Q {f_1, \dots, f_k, \deriveright a f_1, \dots, \deriveright a f_k}_{\infiltration}.
    \end{align*}
    By construction, $\deriveright a f = \deriveright a f_1$ is in $B$.
    It remains to show that $B$ is $\deriveleft {}$-closed.
    Since $A$ is $\deriveleft {}$-closed, it suffices to show this for the new generators.
    So consider $\deriveleft b \deriveright a f_i$
    for an arbitrary some $b \in \Sigma$ and $1 \leq i \leq k$.
    By~\cref{lem:derive left right commutativity}, left and right derivatives commute,
    so we have $\deriveleft b \deriveright a f_i = \deriveright a \deriveleft b f_i$.
    Since $A$ is $\deriveleft {}$-closed, $\deriveleft b f_i$ is in $A$,
    and thus we can write $\deriveleft b f_i = p_i^b(f_1, \dots, f_k)$
    for some polynomial $p_i^a$ (where product is interpreted as infiltration product).
    Thanks to linearity of the right derivative operator and~\cref{lem:derive right infiltration},
    a proof by structural induction on polynomials
    shows that $\deriveright a (p_i^a(f_1, \dots, f_k))$ can be written as a polynomial expression
    $q_i^{a, b}(f_1, \dots, f_k, \deriveright a f_1, \dots, \deriveright a f_k)$.
    We thus have, as required,
    \begin{align*}
        \deriveleft b \deriveright a f_i
            = \deriveright a \deriveleft b f_i
            = q_i^{a, b}(f_1, \dots, f_k, \deriveright a f_1, \dots, \deriveright a f_k)
        \in B.
        \qedhere
    \end{align*}
\end{proof}

\begin{theorem}
    \label{thm:infiltration finite - effective prevariety}
    The class of infiltration-finite series is an effective prevariety.
\end{theorem}
\begin{proof}
    Closure property~\cref{prevariety:A} and closure under left derivatives
    have been established in~\cref{lem:infiltration closure}.
    Closure property~\cref{prevariety:B}
    has been established in~\cref{lem:infiltration closure - right derivative}.
    We will establish decidability of equality in~\cref{thm:infiltration finite - decidable equality}.
\end{proof}

We conclude this section by showing that
infiltration-finite series over disjoint alphabets are closed under shuffle product.
This will be used in~\cref{lem:infiltration commutativity hard}
to show that the commutativity problem generalises the equality problem.
\begin{restatable}{lemma}{lemInfiltrationShuffle}
    \label{lem:infiltration shuffle}
    Let $\Sigma, \Gamma$ be two finite and disjoint alphabets $\Sigma \cap \Gamma = \emptyset$.
    If $f \in \series \Q \Sigma$ and $g \in \series \Q \Gamma$ are infiltration finite,
    then $f \shuffle g$ is infiltration finite.
\end{restatable}
\noindent
The proof is different from the case of Hadamard-finite series~(\cref{lem:Hadamard shuffle}),
and relies on the fact that over disjoint alphabets shuffle and infiltration products coincide.
\begin{proof}
    We will use the fact that infiltration and shuffle product coincide on series with disjoint alphabets.
    \begin{claim*}
        For every series $f \in \series \Q \Sigma$ and $g\in \series \Q \Gamma$
        over disjoint alphabets $\Sigma \cap \Gamma = \emptyset$,
        \begin{align}
            \label{eq:shuffle=infiltration}
            f \shuffle g = f \infiltration g.
        \end{align}
    \end{claim*}
    \begin{claimproof}
        We proceed by coinduction.
        By definition we have
        $\coefficient \e {(f \shuffle g)} = f_\e \cdot g_\e = \coefficient \e {(f \infiltration g)}$.
        Recall the coinductive definition of the two products:
        \begin{align*}
            \deriveleft a (f \shuffle g)
                &= \deriveleft a f \shuffle g + f \shuffle \deriveleft a g, \text{ and} \\
            \deriveleft a (f \infiltration g)
                &= \deriveleft a f \infiltration g
                    + f \infiltration \deriveleft a g
                    + \deriveleft a f \infiltration \deriveleft a g.
        \end{align*}
        There are two cases to consider.
        \begin{enumerate}
            \item In the first case, assume $a \in \Sigma$ and $a \not\in \Gamma$.
            We have $\deriveleft a g = \zero$, so the two equations reduce to
            \begin{align*}
                \deriveleft a (f \shuffle g)
                    = \deriveleft a f \shuffle g \quad\text{and}\quad
                \deriveleft a (f \infiltration g)
                    = \deriveleft a f \infiltration g,
            \end{align*}
            and we can conclude since the two \lhs~are equal by the coinductive hypothesis.

            \item The second case, $a \in \Gamma$ and $a \not\in \Sigma$, is dealt with similarly.
        \end{enumerate}
    \end{claimproof}
    Thanks to the claim and the fact that infiltration-finite series
    are closed under infiltration product (\cref{lem:infiltration closure}),
    we obtain the statement of the lemma.
\end{proof}

\subsection{Equality problem}

The equality problem for infiltration-finite series is decidable.
This can be seen as a generalisation from univariate sequences (\aka~\emph{streams}) $\N \to \Q$
to multivariate series $\series \Q \Sigma$
of the decidability result \cite[Theorem 4.2]{BorealeCollodiGorla:ACMTCL:2024}
instantiated to the stream infiltration product.
\begin{restatable}{theorem}{thmEqInfiltrationDec}
    \label{thm:infiltration finite - decidable equality}
    The equality problem is decidable for infiltration-finite series.
\end{restatable}

To this end, we adapt a classic technique based on
chains of polynomial ideals and Hilbert's finite basis theorem;
\cf~analogous algorithms for polynomial automata~\cite{BenediktDuffSharadWorrell:PolyAut:2017},
shuffle-finite series~\cite{Clemente:CONCUR:2024},
and univariate infiltration-finite series~\cite{BorealeGorla:CONCUR:2021,BorealeCollodiGorla:ACMTCL:2024}.

Fix an infiltration-finite series $f \in \series \Q \Sigma$
recognised by an infiltration automaton $A = \tuple{\Sigma, X, F, \Delta}$.
Since equality reduces to zeroness, we show how to decide $\sem {X_1} = \zero$.
For a length $n \in \N$,
consider the polynomial ideal $I_n$
generated by all configurations reachable from $X_1$
by reading words of length $\leq n$:
\begin{align*}
    I_n := \idealof {\Delta_w X_1} {w \in \Sigma^{\leq n}}.
\end{align*}
(Recall that a \emph{polynomial ideal} is a set of polynomials $I \subseteq \poly \Q X$
\st~$I + I \subseteq I$ and $I \cdot \poly \Q X \subseteq I$.
For a set of polynomials $P$, we write $\ideal P$ for the smallest polynomial ideal containing $P$.
We refer to \cite{CoxLittleOShea:Ideals:2015} for more details on algebraic geometry.)
The following lemma describes the relevant properties of these ideals.
\begin{lemma}
    \label{lem:infiltration ideals}
    \begin{enumerate}
        \item $I_n \subseteq I_{n+1}$.
        \item $\Delta_a I_n \subseteq I_{n+1}$.
        \item $I_{n+1} = I_n + \idealof {\Delta_a I_n} {a \in \Sigma}$.
        \item $I_n = I_{n+1}$ implies $I_n = I_{n+1} = I_{n+2} = \cdots$.
    \end{enumerate}
\end{lemma}
\begin{proof}
    The first point holds by definition.

    For the second point, let $\alpha \in \Delta_a I_n$.
    Then $\alpha$ is of the form
    \begin{align*}
        \alpha
        &= \Delta_a \sum_{i = 1}^m \beta_i \cdot \Delta_{w_i} X_1 = \\
        &= \sum_{i = 1}^m \Delta_a (\beta_i \cdot \Delta_{w_i} X_1) = \\
        &= \sum_{i = 1}^m \left(\Delta_a \beta_i \cdot \Delta_{w_i} X_1
            + S_a \beta_i \cdot \Delta_a \Delta_{w_i} X_1 \right) = \\
        &= \sum_{i = 1}^m \left(\Delta_a \beta_i \cdot \underbrace{\Delta_{w_i} X_1}_{I_n}
            + S_a \beta_i \cdot \underbrace{\Delta_{a \cdot w_i} X_1}_{I_{n+1}}\right),
    \end{align*}
    which is clearly in $I_{n+1}$.

    We now consider the third point.
    The ``$\supseteq$'' inclusion holds by the first two points
    and the fact that $I_{n+1}$ is an ideal.
    For the other inclusion, let $\alpha \in I_{n+1}$.
    We can write $\alpha$ separating words of length exactly $n+1$ from the rest, and we have
    \begin{align*}
        \alpha
            &= \underbrace \beta_{I_n}
            + \sum_{w \in \Sigma^{n+1}} \beta_w \cdot \Delta_w X_1 = \\
            &= \beta + \sum_{a \in \Sigma} \sum_{w \in \Sigma^n} \beta_{w \cdot a} \cdot \Delta_{w \cdot a} X_1 = \\
            &= \beta + \sum_{a \in \Sigma} \sum_{w \in \Sigma^n} \beta_{w \cdot a} \cdot \Delta_a \underbrace{\Delta_w X_1}_{I_n},
    \end{align*}
    where the latter quantity is clearly in the ideal generated by $\Delta_a I_n$.

    The last point follows from the previous one
    since we have shown that $I_{n+1}$ is a function of $I_n$.
\end{proof}

Now consider the chain of ideals
\begin{align}
    \label{eq:infiltration ideal chain}
    I_0 \subseteq I_1 \subseteq \cdots \subseteq \poly \Q X.
\end{align}
By \emph{Hilbert finite basis theorem}~\cite[Theorem 4, §5, Ch. 2]{CoxLittleOShea:Ideals:2015},
there is $N \in \N$ \st~the chain stabilises $I_N = I_{N+1} = \cdots$.
Thanks to the last point of~\cref{lem:infiltration ideals}, such an $N$ is computable:
Indeed, $N$ can be taken to be the smallest $n \in \N$ \st~$I_n = I_{n+1}$,
and ideal equality can be decided by checking that the generators of $I_{n+1}$ are all in $I_n$.
The latter test is an instance of the \emph{ideal membership problem},
which is decidable (in exponential space)~\cite{Mayr:STACS:1989}.

The following property is easy to check.
\begin{lemma}
    \label{lem:infiltration zeroness characterisation}
    Let $N \in \N$ be the stabilisation index
    of the ideal chain~\cref{eq:infiltration ideal chain}.
    $\sem {X_1} = \zero$ if, and only if,
    $\coefficient w \sem {X_1} = 0$ for all $w \in \Sigma^{\leq N}$.
\end{lemma}
\begin{proof}
    One direction is trivial.
    For the other direction we have,
    for every word $w \in \Sigma^*$,
    \begin{align*}
        \deriveleft w \sem {X_1}
        &= \sem {\Delta_w X_1}
        = \sem {\sum_{u \in \Sigma^{\leq N}} \beta_u \cdot \Delta_u X_1} = \\
        &= \sum_{u \in \Sigma^{\leq N}} \sem {\beta_u} \infiltration \sem {\Delta_u X_1},
    \end{align*}
    where we have used the fact that $\Delta_w X_1$ is in $I_N$.
    By taking constant terms on both sides, we have
    \begin{align*}
        \coefficient w \sem {X_1}
        &= \coefficient \e \deriveleft w \sem {X_1}
        = \sum_{u \in \Sigma^{\leq N}} \coefficient \e \sem {\beta_u} \cdot \coefficient \e \sem {\Delta_u X_1} = \\
        &= \sum_{u \in \Sigma^{\leq N}} \coefficient \e \sem {\beta_u} \cdot \coefficient u \sem {X_1}.
    \end{align*}
    But $u$ has length $\leq N$, and thus $\coefficient u \sem {X_1} = 0$,
    implying $\coefficient w \sem {X_1} = 0$, as required.
\end{proof}
\begin{proof}[Proof of~\cref{thm:infiltration finite - decidable equality}]
    By using the effective closure properties of infiltration-finite series,
    we reduce equality $f = g$ to zeroness $f - g = \zero$.
    So let $f$ be an infiltration-finite series,
    recognised by an infiltration automaton $A$ as above.
    Compute the stabilisation index $N \in \N$ of the ideal chain~\cref{eq:infiltration ideal chain}.
    Thanks to~\cref{lem:infiltration zeroness characterisation},
    we can decide $f = \zero$
    by enumerating all words $w$ of length $\leq N$
    and checking $\coefficient w f = 0$.
    The latter test is performed by applying the semantics of the infiltration automaton.
\end{proof}

\subsection{Commutativity problem}

Having established that infiltration-finite series
are an effective prevariety by~\cref{thm:infiltration finite - effective prevariety},
it follows from \cref{thm:commutativity for effective prevarieties}
that commutativity is decidable for them.

\begin{theorem}
    The commutativity problem is decidable for infiltration-finite series.
\end{theorem}

In fact, commutativity is polynomial time inter-reducible with the equivalence problem.
\begin{lemma}
    \label{lem:infiltration commutativity hard}
    The commutativity problem for infiltration-finite series
    is at least as hard as the equivalence problem, under polynomial time reductions.
\end{lemma}
\begin{proof}
    Let $f \in \series \Q \Sigma$ be an infiltration-finite series
    and consider two fresh input symbols $a, b \not\in \Sigma$.
    Thanks to~\cref{lem:shuffle reflects commutativity},
    $g := ab \shuffle f$ is commutative if, and only if, $f = \zero$.
    Moreover, since $ab$ is infiltration finite (even rational),
    by~\cref{lem:infiltration shuffle} the series $g$ is infiltration finite.
    Additionally, a finite representation for $g$ can be constructed in polynomial time
    from a finite representation of $f$.
    We have reduced checking whether $f = \zero$
    to checking whether $g$ is commutative.
\end{proof}
\section{Effective varieties and commutativity}
\label{app:varieties}
In this section we provide more details about~\cref{sec:effective algebraic varieties}.
Given a set of polynomials $P \subseteq \poly \Q X$,
their \emph{zero set} $V(P) \subseteq \C^k$ is the set of their common zeroes,
\begin{align*}
    V(P) := \setof{x \in \C^k}{\forall p \in P : p(x) = 0}.
\end{align*}
Notice that $V(P)$ does not change if we close $P$ under addition, and under product with $\poly \Q X$.
This gives rise to the notion of a \emph{polynomial ideal}~\cite[§4, Ch.~1]{CoxLittleOShea:Ideals:2015},
which is a set of polynomials $I \subseteq \poly \Q X$
\st~$I + I \subseteq I$ and $I \cdot \poly \Q X \subseteq I$.
For a set of polynomials $P$, let $\ideal P$ be the smallest polynomial ideal containing $P$.
Since $V(P) = V(\ideal P)$, we can consider polynomial ideals from now on.

Recall that $\FF = V(P)$ is the zero set of the polynomials $P$ from~\cref{eq:commutativity polynomials}.
Thus, $\FF = V(\ideal P)$,
which, by definition, means that $\FF$ is an \emph{algebraic variety}~\cite[§2, Ch.~1]{CoxLittleOShea:Ideals:2015}.
By \emph{Hilbert's finite basis theorem}~\cite[Theorem 4, §5, Ch. 2]{CoxLittleOShea:Ideals:2015},
$\ideal P$ is finitely generated, however in general there is no computable bound on the number of generators.
In our case however, we can compute finitely many generators for $\ideal P$.
Thanks to the proof leading to decidability of commutativity,
there is a computable bound $N \in \N$ \st\-
the ideal $\ideal P$ is already generated by $P_N \subseteq P$.
In other words, we have $\ideal P = \ideal {P_N}$,
where
\begin{align}
    \label{eq:finitely many commutativity polynomials}
    P_N := \setof {\Delta_u \alpha - \Delta_v \alpha \in \poly \Q X}{u, v \in \Sigma^{\leq N}, u \sim v}.
\end{align}
Notice that $P_N$ is finite and computable.
The number $N$ can be extracted from the proof of decidability of commutativity.
For Hadamard automata, it is Ackermannian,
and for shuffle automata it is doubly exponential.

The existential variant of the commutativity problem
amounts to decide whether there exists an output function $F\in\C^k$
giving rise to a commutative semantics.
In other words, this asks whether $V(\ideal {P_N}) \neq \emptyset$.
By \emph{Hilbert's weak Nullstellensatz}~\cite[Theorem 1, §1, Ch.~4]{CoxLittleOShea:Ideals:2015},
this is the same as checking whether $1 \not \in \ideal {P_N}$,
which can be decided with Gröbner bases~\cite[Ch.~2]{CoxLittleOShea:Ideals:2015}.

The universal variant of commutativity
amounts to decide whether all output functions $F\in\C^k$ give rise to a commutative semantics.
This is the same as whether $V(\ideal {P_N}) = \C^k$,
and thus it means that $P_N$ contains no nontrivial constraint,
i.e., $P_N = \ideal {P_N} = \set 0$.

\end{document}